\documentclass[a4paper,UKenglish]{article}

\usepackage{microtype}

\usepackage[utf8]{inputenc}
\usepackage[british]{babel}
\usepackage{a4}
\usepackage{a4wide}
\usepackage[a4paper,colorlinks,pagebackref,linkcolor=blue,citecolor=red,urlcolor=cyan]{hyperref}
\usepackage{amsmath,amssymb,stmaryrd,amsfonts}

\usepackage{Common/theorems}
\usepackage{Common/prooftree}
\usepackage{Common/appendix}
\capturecounter{theorem}
\capturecounter{definition}

\title{A simple presentation of the effective topos}

\author{Alexis Bernadet$^{1}$, Stéphane Graham-Lengrand$^{1,2}$\\
  $^{1}$\'Ecole Polytechnique, France\\
  $^{2}$CNRS, France}

\long\def\ignore#1{\relax}

\newcommand\struto[1][15pt]{{\raise #1 \hbox{\strut}}}%
\newcommand\strutb[1][15pt]{{\raise-#1 \hbox{\strut}}}%

\newcommand\pushright[1]{{              
    \parfillskip=0pt            
    \widowpenalty=1000000000         
    \displaywidowpenalty=10000  
    \finalhyphendemerits=0      
    \leavevmode                 
    \unskip                     
    \nobreak                    
    \hfil                       
    \penalty50                  
    \hskip.2em                  
    \null                       
    \hfill                      
    {#1}                        
    \par}}                      

\newcommand\olditem{}
\newcommand\olditemize{}
\newcommand\oldenditemize{}
\newcommand\oldenumerate{}
\newcommand\oldendenumerate{}
\let\olditem\item
\let\olditemize\itemize
\let\oldenditemize\enditemize
\let\oldenumerate\enumerate
\let\oldendenumerate\endenumerate

\newcommand\myitem{}
\makeatletter
\def\myitem{\@ifnextchar[\@myitemwith\@myitemwithout}
\long\def\@myitemwith[#1]{\olditem[{#1}]\unskip}
\long\def\@myitemwithout{\olditem\unskip}
\makeatother

\renewenvironment{itemize}[1][0]{%
  \def\item{\removelastskip\myitem}%
  \removelastskip\olditemize\removelastskip}
{\removelastskip\oldenditemize\removelastskip%
  \def\item\olditem%
}

\renewenvironment{enumerate}[1][0]{%
  \def\item{\removelastskip\myitem}%
  \removelastskip\oldenumerate\removelastskip}
{\removelastskip\oldendenumerate\removelastskip%
  \def\item\olditem%
}

\newcommand\mybox[1]{\fbox{\vbox{#1}}}
\renewcommand\[[1][3]{\par\removelastskip\vskip#1pt\vbox\bgroup\hrule height0pt\vfil\hbox to\hsize\bgroup\hfil\(}
\renewcommand\][1][3]{\)\hfil\egroup\vfil\hrule height0pt\egroup\vskip#1pt\nointerlineskip\noindent}

\newbox\columnsbox
\newbox\tmpbox
\newdimen\columnsheight
\newdimen\columnwidth
\newdimen\remainingwidth
\newdimen\textwidthsave
\def\mycolumnsheight{}

\newcommand\columns[1]{%
  \def\mycolumnsheight{}%
  \setlength\remainingwidth\textwidth%
  \setbox\columnsbox=\vbox\bgroup\vskip0pt\vfil\hbox to\textwidth\bgroup#1\egroup\vfil\egroup%
  \columnsheight=\ht\columnsbox%
  \def\mycolumnsheight{to\columnsheight}%
  \hrule height 0pt\vtop{\hbox to\wd\columnsbox\bgroup#1\egroup}%
}

\makeatletter

\def\commonpart{%
  \setlength\columnwidth{\wd\tmpbox}%
  \vtop{\vskip0pt\hbox to\columnwidth{{\box\tmpbox}}}%
  \advance\remainingwidth-\columnwidth%
  \setlength\textwidth\textwidthsave%
  \hsize\textwidthsave%
}
\def\column{\unskip\setlength\textwidthsave\textwidth\@ifnextchar[\@columnwith\@columnwithout}
\long\def\@columnwith[#1]#2{%
  \def\newhsize{#1\dimexpr\textwidth\relax}%
  \hsize\newhsize%
  \ifdim\hsize<0.1pt\hsize\remainingwidth\fi%
  \setlength\textwidth\hsize%
  \setbox\tmpbox=\hbox to\hsize\bgroup\hfil\vtop\mycolumnsheight{\vskip0pt#2\vskip0pt}\hfil\egroup%
  \commonpart%
}
\long\def\@columnwithout#1{%
  \hsize\remainingwidth%
  \setlength\textwidth\hsize%
  \setbox\tmpbox=\hbox\bgroup\vtop\mycolumnsheight{\vskip0pt#1\vskip0pt}\egroup%
  \commonpart%
}
\makeatother

\newcommand{\eqdef}{:=\ }

\renewcommand\l{\lambda}

\newcommand\mathFomega{F_\omega}
\newcommand\Fomega{\ifmmode\mathFomega\else$\mathFomega$\fi}
\newcommand\mathFomegaC{F_\omega^{\mathcal C}}
\newcommand\FomegaC{\ifmmode\mathFomegaC\else$\mathFomegaC$\fi}
\newcommand\mathDNE{\mathrm{DNE}}
\newcommand\DNE{\ifmmode\mathDNE\else$\mathDNE$\fi}

\newcommand{\ie}{i.e.~}

\ifdefined\url
\else
\def\url[#1]#2{\texttt{#2}}
\fi

\let\oldurl\url

\makeatletter
\ifdefined\href
\def\myurl{\@ifnextchar[\@myurlwith\@myurlwithout}
\long\def\@myurlwith[#1]#2{\mbox{\href{#2}{#1}}}
\long\def\@myurlwithout#1{\mbox{\href{#1}{#1}}}
\else
\def\myurl{\@ifnextchar[\@myurlwith\@myurlwithout}
\long\def\@myurlwith[#1]#2{\mbox{\oldurl[#1]{#2}}}
\long\def\@myurlwithout#1{\mbox{\oldurl{#1}}}
\fi
\makeatother

\def\url{\myurl}

\newcommand\monthdisplay[1]{\unskip}

\newcommand{\eqeff}[3]{| #2 =_{#1} #3|}
\newcommand{\eqeffEl}[3]{\eqeff {El(#1)} {#2} {#3}}

\newcommand{\Nat}[0]{\mathbb{N}}

\newcommand\sep\mid

\begin{document}

\maketitle

\begin{abstract}
  We propose for the Effective Topos an alternative construction: a realisability framework composed of two levels of abstraction.

  This construction simplifies the proof that the Effective Topos is a topos (equipped with natural numbers), which is the main issue that this paper addresses.
  In this our work can be compared to Frey's monadic tripos-to-topos construction.

  However, no topos theory or even category theory is here required for the construction of the framework itself, which provides a semantics for higher-order type theories, supporting extensional equalities and the axiom of unique choice.
\end{abstract}


\section{Introduction}

Topos Theory~\cite{LawvereFW:eletcs} can be used to abstract models of
higher-order logic the same way Heyting algebras can be used for propositional
logic. Contrary to set
theories such as ZF, in Topos Theory, the expressions are typed.

A topos is a Cartesian closed category that satisfies certain properties. Even
if the definition is simple, a topos describes an entire mathematical universe
of high-order logic. The category of sets is a topos. So, when dealing with the
internal logic of a topos, if we prove a result that is true in every topos,
then this result is true in Set Theory.
Therefore, the internal logic of a topos is not too "exotic".

In a topos, it is interesting to consider whether it has the following
properties:
\begin{itemize}
\item The internal logic of the topos has the law of excluded middle
\item Every morphism of the topos is ``computable''
\item The topos has \emph{an object of natural integers} (definition of Lawvere)
\end{itemize}
With these three properties we can prove that the halting problem is computable
which is absurd.
Hence, a (non-degenerated) topos cannot have these three properties at the same
time, so we have to use a different topos according to which high roder logic
we want.
\begin{itemize}
\item If we have the first and the third properties, we have classical logic
with arithmetics:
the category of sets is such a topos.
\item If we have the first and second properties but not the third one we have
finite logic: the category of finite sets is such a topos.
\item If we have the second and third properties but not the first one we have
  intuitionistic logic. There are several topos that satisfy these properties.
  One of the most famous ones is Hyland's effective topos~\cite{HylandJ:efft,PhoaW:fttetms}, which can
be seen as the universe of realisability.
  Topos satisfying the second and third properties are the most interesting ones
for computer science because they ensure that a programming language based on
the Topos Theory can be given a constructive semantics.

In such a programming language, we can only write functions that terminate,
as in proof assistants like Coq, so the language cannot be Turing-complete.
The main advantage of having a programming language based on Topos Theory over
more usual intuitionistic systems such as Martin-Loef type theory is the notion
of equality:
it is extensional, has proof-irrelevance, and allows the axiom of unique choice.
\end{itemize}
It is quite simple to check that the category of sets and the category of
finite sets are topos.
However, proving that the effective topos is indeed a topos is much harder
\cite{HylandJ:efft}.
This is mainly do to the fact that the notion of morphism in this category is
not intuitive.
The proof that the effective topos is a topos can be generalised with the Tripos
Theory~\cite{PittsAM:thet}, but this does not simplify the structure of the
proof.

In this paper we present an alternative and simplier construction of a topos
that turns out to be equivalent to Hyland's effective topos:
This construction is based on a realisability framework with two levels of
abstraction: a \emph{low level}, comprising all the objects of the topos,
and a \emph{high level} used to define the morphisms.

Moreover, the high level
\begin{itemize}
\item identifies the properties that are needed to prove
that the framework froms a topos, as simply as proving that the category of sets
is a topos;

\item can be directly used as a model of higher-order intuitionistic systems:
Building such a semantics within the high level of the framework relies on the
properties we prove to show that the full framework forms a topos.

\end{itemize}
This work can be compared to the monadic construction from Tripos to Topos
\cite{Frey2Cat2011} but the constructions and properties of the framework
does not require knowledge about category theory nor Topos Theory.

To our knowledge, the construction of this framework and this presentation of
the effective Topos is new.

In Section 2, we define the core of the framework and prove its basic
properties, especially how the results on the low level can be lifted to the
high level.
In Section 3, we give our own definition of an effective topos, we enrich the
framework and prove that our effective topos is indeed a topos with an object of
natural integers.


\section{Presentation and general tools of the framework}

In this section we are going to define the core and basic tools of our
realisability framework.
The main reason we choose to base our framework on realisability is to be able
to do program extraction (see Theorem \ref{th:ProgramExtraction}).
In Section \ref{sec:Realisability} we define the realisability part of the
framework.
In Section \ref{sec:Eff} we define the effective sets which will be the objects
of the effective topos.
In Section \ref{sec:Elements} we construct the high level part of the framework
which is needed if we want the axiom of unique choice and the fact that the
category that we construct is a topos (see Appendix \ref{sec:Naive}).
In Section \ref{sec:SP} we prove that a property that is true at the low
level is also true at the high level, which is useful, because most of the
properties we need for the framework are high-level properties.
In Section \ref{sec:Function} we define what a function in the framework is and
how to build a high-level function from a low-level one.
This is useful because most of the functions we need in our framework are
high-level ones.

\subsection{Realisability}

\label{sec:Realisability}

Our framework is based on a notion of realisability that interprets formulae
as sets of ``proofs''.
As in Hyland's construction of the effective topos, we shall use integers to
represent ``proofs'', but it could in fact be done with other well-known
structures such as $\l$-terms.

\begin{notation}
If $n$ and $m$ integers, then we can code $(n, m)$ by an integer and we write it
$<n, m>$.

If $e$ is an integer that codes a partial recursive function, then we write
$\varphi_e$ this function. And for all $n$, we write $\varphi_e(n) \downarrow$
if $\varphi_e$ is defined in $n$ and we write $\varphi_e(n)$ the image of $n$ by
$\varphi_e$.
\begin{itemize}

\item We write $\varphi_e(n) \downarrow = m$ for $\varphi_e(n) \downarrow$ and
$\varphi_e(n) = m$.

\item We write $\varphi_e(n) \downarrow \in F$ for $\varphi_e(n) \downarrow$ and
$\varphi_e(n) \in F$

\end{itemize}
We write $Prop = P(\Nat)$ with $\Nat$ the set of integers.
\end{notation}
The definition of logical operators is inspired by Heyting semantics.
\begin{definition}[Logical operators]\strut

If $F, G \in Prop$ and $H \in X \rightarrow Prop$ with $X$ a set, then we write:
\[
\begin{array}{l@{\eqdef}l}
\top & \Nat \\
\bot & \emptyset \\
F \wedge G & \{ < n, m > \mid n \in F ~ m \in G \}\\
F \Rightarrow G & \{ e \mid \varphi_e ~ exists \wedge \forall n \in F,
\varphi_e(n) \downarrow \in G\}\\
\forall x \in X, H(x) & \bigcap_{x \in X} H(x)\\
\exists x \in X, H(x) & \bigcup_{x \in X} H(x)
\end{array}
\]
$F\Leftrightarrow G$ is an abbreviation for $(F \Rightarrow G) \wedge
(G \Rightarrow F)$, and we write $\vDash F$ if and only if
$F \neq \emptyset$.
\end{definition}
\begin{theorem}[Capturing intuitionistic provability]\strut

The notion of inhabitation denoted $\vDash F$ with logic operators defined
above admits the rules of deduction of intuitionistic first-order logic.
\end{theorem}
In the rest of the paper, we use this property implicitly to derive
inhabitation results from simple intuitionist reasoning.
For instance to prove $\vDash F \Rightarrow G$ we suppose $F$ and then
prove $G$: We admit that the reader can transform a intuitionistic
first-order proof to an integer representing a program (via Curry-Howard).
We can do this because there is no ambiguity between elements of $Prop$ and
real mathematical formulae.
The only theorems where we need to explicitly manipulate the integers as proofs
are Lemma \ref{lem:NProperties} and Theorem \ref{th:ProgramExtraction}.
\remark{Had we chosen $\l$-terms to represent proofs, every result of the form
$\vDash F$ that we prove in this paper would be such that $F$ is inhabited by a
$\lambda$-term typable in an extension of $F_\omega$.
\footnote{We would need product
sorts, and some notions of arithmetic -see the required notions in Section 3.}
}

\subsection{Effective sets}

\label{sec:Eff}

The definition of effective sets is the same as the usual
one~\cite{HylandJ:efft}.
It can be seen as a set with a partial equivalence relation
(in the internal logic of the realisability).

\begin{definition}[Effective sets]

$X = (|X|, |.=_X.|)$ is an effective set if and only if :
\begin{itemize}

\item $|X|$ is a set.

\item $|.=_X.| \in |X| \times |X| \rightarrow Prop$, and we write
$\eqeff X x y$ for $|.=_X.|(x, y)$.

\item $\vDash \forall x, y \in |X|, \eqeff X x y \Rightarrow \eqeff X y x$
(Symmetry)

\item $\vDash \forall x, y, z \in |X$, $\eqeff X x y \Rightarrow \eqeff X y z
\Rightarrow \eqeff X x z$ (Transitivity)

\end{itemize}
\end{definition}
The main reason we do no require reflexivity in the definition is that a proof
of $\eqeff X x x$ may contain information.
See Section \ref{sec:Nat} for example.
\begin{remark}
Assume $X$ is an effective set.
Then by symetry and transitivity we have\\
$\vDash \forall x, y \in |X|, \eqeff X x y \Rightarrow
(\eqeff X x x \wedge \eqeff X y y)$
\end{remark}

\begin{notation}[Quantification over an effective set]\strut

Assume $X$ is an effective set and $F \in |X| \rightarrow Prop$.
We write:
\begin{itemize}

\item $\forall x \in X, F(x)$ for $\forall x \in |X|,
\eqeff X x x \Rightarrow F(x)$

\item $\exists x \in X, F(x)$ for $\exists x \in |X|,
\eqeff X x x \wedge F(x)$.

\end{itemize}
\end{notation}
This notation will shorten many properties and will make them more readable.

\subsection{Elements}
\label{sec:Elements}
Contrary to the usual presentation of the effective topos, we use the
notion of $El(X)$ as a core feature of the framework.
\begin{definition}[Elements of an effective set]\strut

If $X$ is an effective set, we write $El(X)$ the effective set defined by:
\[\begin{array}{lcll}
|El(X)| &\eqdef& |X| \rightarrow Prop & \\
\eqeffEl X u v &\eqdef&
(\forall x, x' \in |X|, u(x) \Rightarrow \eqeff X x {x'} \Rightarrow u(x'))
& \text{(Stability)}\\
&\wedge &(\forall x, x' \in |X|, u(x) \Rightarrow u(x') \Rightarrow
\eqeff X x {x'}) & \text{(Unicity)} \\
&\wedge &(\exists x \in |X|, u(x)) & \text{(Existence)} \\
&\wedge &(\forall x \in |X|, u(x) \Leftrightarrow v(x)) & \text{(Equivalence)}
\end{array}\]
It is straightforward to show that $El(X)$ is indeed an effective set.

$El(X)$ can be seen as the type of singletons included in $X$.

\end{definition}
\begin{remark}
Assume $X$ is an effective set.
Then by uncity of $u$ we have: \\
$\vDash \forall u, v \in |El(X)|, x \in |X|,
\eqeffEl X u v \Rightarrow u(x) \Rightarrow \eqeff X x x$
\end{remark}


Manipulating $x \in |X|$ is considered low-level
and manipulting $u \in |El(X)|$ without knowing that
$|El(X)| = |X| \rightarrow Prop$ is considered high-level.
With enough high-level theorems, it is possible to prove that the category
we build is a topos by adapting the proof that the category of sets is a topos.

To achieve this we need a systematic way to lift structures and properties from
the low level to the high level.
\begin{definition}[Injection from the low level to the high level]
If $X$ is an effective set and $x \in |X|$, we define $el_X(x)$ as follows:
For all, $y \in |X|$, $el_X(x)(y) \eqdef \eqeff X x y$.
\end{definition}

\begin{toappendix}

\appendixbeyond 0

\begin{lemma}[Basic relation between the low level and the high level]\strut

\label{lem:LowLevelHighLevel}

If $X$ is an effective set then:
\begin{enumerate}

\item $\vDash \forall u, v \in |El(X)|,
\eqeffEl X u u \Rightarrow (\forall x \in |X|, u(x) \Leftrightarrow v(x))
\Rightarrow \eqeffEl X u v$

\item $\vDash \forall u, v \in El(X),
(\forall x \in |X|, u(x) \Rightarrow v(x)) \Rightarrow \eqeffEl X u v$

\item $\vDash \forall u \in |El(X)|,
\eqeffEl X u {el_X(x)} \Leftrightarrow (\eqeffEl X u u \wedge u(x))$

\item $\vDash \forall x, y \in |X|, \eqeff X x y \Leftrightarrow
\eqeffEl X {el_X(x)} {el_X(y)}$

\item $\vDash \forall u \in |El(X)|,
\eqeffEl X u u \Leftrightarrow \exists x \in |X|, \eqeff X u {el_X(x)}$

\end{enumerate}

\end{lemma}

\end{toappendix}

\begin{toappendix}[\begin{proof}

Straightforward. See Appendix \thisappendix.

\end{proof}]

\begin{proof}

\begin{enumerate}

\item $\eqeffEl X u u$.
Hence, all the properties of $\eqeffEl X u v$ that do not talk about $v$
are true (stability, unicity and existence).
Therefore, with the equivalence of $u$ and $v$ we have $\eqeffEl X u v$.

\item From the previous point, we only have to prove that if
$v(x)$ then $u(x)$ for all $x \in |X|$:
Assume we have $v(x)$.
From $\eqeffEl X u u$, there exists $y$ such that $u(y)$.
By hypothesis, we have $v(y)$.
From $\eqeffEl X v v$ we have $\eqeff X y x$.
Therefore, from $\eqeffEl X u u$, we have $u(x)$.

\item

\begin{itemize}

\item If $\eqeffEl X u {el_X(x)}$ then by symmetry and transitivity
we have $\eqeffEl X u u$.
Also, there exists $y \in |X|$ such that $u(y)$.
So we have $el_X(x)(y)$ which means $\eqeff X x y$.
Therefore, $\eqeff X y x$ and $u(x)$ (stability of $u$).

\item Assume we have $\eqeffEl X u u$ and $u(x)$.
Let $y$ such that $u(y)$, so we have $\eqeff X x y$ (by unicity).
Let $y$ such that $\eqeff X x y$, so we have $u(y)$ (stability).
Hence, for all $y \in |X|$, $u(y)$ if and only if $el_X(x)(y)$.
Therefore, from the first point we have $\eqeffEl X u {el_X(x)}$.

\end{itemize}

\item

\begin{itemize}

\item Assume $\eqeff X x y$.

\begin{itemize}

\item (Stability) If $\eqeff X x z$ and $\eqeff X z {z'}$ then
$\eqeff X x {z'}$.
\item (Unicity) If $\eqeff X x z$ and $\eqeff X x {z'}$ then $\eqeff X z {z'}$.
\item (Existence) $\eqeff X x y$ so there exists $z$ such that $\eqeff X x z$.
\item (Equivalence) $\eqeff X x z$ if and only if $\eqeff X y z$.

\end{itemize}

Therefore $\eqeffEl X {el_X(x)} {el_X(y)}$.

\item Assume $\eqeffEl X {el_X(x)} {el_X(y)}$.

So there exists $z$ such that $el_X(x)(z)$.
Hence $el_X(y)(z)$.
Therefore we have $\eqeff X x z$ and $\eqeff X y z$.
Hence $\eqeff X x y$.

\end{itemize}

\item

\begin{itemize}

\item Assume $\eqeffEl X u u$.
So there exists $x$ such that $u(x)$.
Therefore we have $\eqeffEl X u {el_X(x)}$ from the third point.

\item Assume there exists $x$ such that $\eqeffEl X u {el_X(x)}$.
From the third point we have $\eqeffEl X u u$.

\end{itemize}

\end{enumerate}

\end{proof}

\end{toappendix}

\subsection{Stable predicates}

\label{sec:SP}

In this section we prove that properties that are true at the low-level are
also true at the high-level.
However, this is true only if the property is stable by equality.
Hence we gave the following definition:
\begin{definition}[Stable predicates]

Assume $X_1$, \ldots, $X_n$ are effective sets.

A \emph{stable predicate} on $(X_1, \ldots, X_n)$ is a
$F \in (|X_1| \times \ldots |X_n|) \rightarrow Prop$ such that:

\[\begin{array}{l}
\vDash \forall x_1 \in |X_1|, \ldots, x_n \in |X_n|,
x_1' \in |X_1|, \ldots, x_n' \in |X_n|,\\
\qquad\qquad
\eqeff {X_1} {x_1} {x_1'} \Rightarrow \ldots \Rightarrow \eqeff {X_n} {x_n}
{x_n'} \Rightarrow F(x_1, \ldots, x_n) \Rightarrow F(x_1', \ldots, x_n')
\end{array}
\]
We write $SP(X_1, \ldots, X_n)$ the set of stable predicates on
$(X_1, \ldots, X_n)$.
\end{definition}
Generally, it is straightforward to prove that a property is stable because we
only manipulate stable predicates and stable functions
(see Section \ref{sec:Function}).

Then we can prove the main goal of this section.
Notice that the lifting does not have to be on all the arguments of the
predicate.
\begin{theorem}[Extension of truth]

\label{th:ExtensionTruth}

Assume $X_1$, \ldots, $X_n$ are effective sets, $k \leq n$, \\
$F \in SP(El(X_1), \ldots, El(X_k), X_{k + 1}, \ldots X_n)$
such that:

\[\vDash \forall x_1 \in X_1, \ldots , x_n \in X_n,
F(el_{X_1}(x_1), \ldots , el_{X_k}(x_k), x_{k + 1}, \ldots x_{n})\]
Then:
\[\vDash \forall u_1 \in El(X_1), \ldots , u_k \in El(X_k),
x_{k + 1} \in X_{k + 1}, \ldots x_n \in X_n,
F(u_1, \ldots , u_k, x_{k + 1}, \ldots x_n)\]
\end{theorem}

\begin{proof}

If $\eqeffEl {X_1} {u_1} {u_1}$, \ldots $\eqeffEl {X_k} {u_k} {u_k}$,
$\eqeff {X_{k + 1}} {x_{k + 1}} {x_{k + 1}}$, \ldots
$\eqeff {X_n} {x_n} {x_n}$,
then by Lemma \ref{lem:LowLevelHighLevel}.5 there exist
$x_1 \in |X_1|$, \ldots , $x_k \in |X_k|$ such that
$\eqeffEl {X_1} {u_1} {el_{X_1}(x_1)}$, \ldots ,
$\eqeffEl {X_k} {u_k} {el_{X_k}(x_k)}$.
Then by symmetry we have $\eqeffEl {X_1} {el_{X_1}(x_1)} {u_1}$, \ldots ,
$\eqeffEl {X_k} {el_{X_k}(x_k)} {u_k}$.
And by transitivity we have
$\eqeffEl {X_1} {el_{X_1}(x_1)} {el_{X_1}(x_1)}$, \ldots ,
$\eqeffEl {X_k} {el_{X_k}(x_k)} {el_{X_k}(x_k)}$.
So by Lemma \ref{lem:LowLevelHighLevel}.4 we have $\eqeff {X_1} {x_1} {x_1}$,
\ldots , $\eqeff {X_k} {x_k} {x_k}$.
Hence by hypothesis we have
$F(el_{X_1}(x_1), \ldots , el_{X_k}(x_k), x_{k + 1}, \ldots x_n)$.
And then by stability we have
$F(u_1, \ldots , u_k, x_{k + 1}, \ldots x_n)$.

\end{proof}
With this theorem we can prove that if a proposition is true at the low level
then it is true at the high level. We just have to check the stability of the
proposition which is usually straightforward.

\subsection{Stable functions}

\label{sec:Function}

By manipulating effective sets, it is natural to be intersted in functions that
are stable with the equality.


\begin{definition}[Stable functions]

Assume $X_1$, \ldots, $X_n$ and $Y$ effective sets.

A \emph{stable function} from $X_1$, \ldots $X_n$ to $Y$, is a
$f \in (|X_1| \times \ldots |X_n|) \rightarrow |Y|$ such that:
\[\begin{array}{l}
\vDash \forall x_1 \in |X_1|, \ldots, x_n \in |X_n|,
x_1' \in |X_1|, \ldots, x_n' \in |X_n|,\\
\qquad\qquad
\eqeff {X_1} {x_1} {x_1'} \Rightarrow \ldots \Rightarrow
\eqeff {X_n} {x_n} {x_n'} \Rightarrow
\eqeff Y {f(x_1, \ldots, x_n)} {f(x_1', \ldots, x_n')}
\end{array}
\]
Such a stable function $f$ is denoted
$f : X_1 \rightarrow \ldots \rightarrow X_n \rightarrow Y$.
\end{definition}
In this paper $n$ will be equal to $1$ or $2$.

We will identify functions by extensionality.

\begin{definition}[Equivalence of functions]\strut

Assume $X_1$, ... $X_n$, and $Y$ are effective sets,
$f : X_1 \rightarrow \ldots X_n \rightarrow Y$ and
$g : X_1 \rightarrow \ldots X_n \rightarrow Y$.

Then we write $f \approx g$ if and only if
\[\vDash \forall x_1 \in X_1, \ldots x_n \in X_n,
\eqeff Y {f(x_1, \ldots, x_n)} {g(x_1, \ldots, x_n)}\]
It is straightforward that $\approx$ is an equivalence relation.
\end{definition}
With the following theorem we are able to write high level functions
by using low level ones.

\begin{definition}[Extension of function]\strut

Assume $X_1$, ... $X_n$, and $Y$ effective sets,
$f : X_1 \rightarrow \ldots X_n \rightarrow El(Y)$ and $k \leq n$.\\
An \emph{extension} of $f$ on the first $k$ arguments is a
$g : El(X_1) \rightarrow \ldots El(X_k) \rightarrow X_{k + 1}
\rightarrow \ldots X_n \rightarrow El(Y)$ such that:
\[\vDash \forall x_1 \in X_1, ... x_n \in X_n,
\eqeffEl Y {g(el_{X_1}(x_1), ..., el_{X_k}(x_k), x_{k + 1}, \ldots, x_n)} 
{f(x_1, ..., x_n)}
\]
\end{definition}

\begin{toappendix}

\appendixbeyond 0

\begin{theorem}[Existence and unicity of extensions of functions]\strut

\label{th:ExistenceUnicityExtensionFun}

Assume $X_1$, ... $X_n$, and $Y$ effective sets,
$f : X_1 \rightarrow \ldots X_n \rightarrow El(Y)$ and $k \leq n$.

We construct
$g \in (|El(X_1)| \times ... \times |El(X_k)| \times
|X_{k + 1}| \times ... \times |X_n|) \rightarrow
|El(Y)|$ defined by:

\[\begin{array}{lll}
g(u_1, \ldots, u_k, x_{k + 1} \ldots x_n)(y) \eqdef &
\exists x_1 \in |X_1|, ..., x_k \in |X_k|, \\
& u_1(x_1) \wedge ... \wedge u_k(x_k) \wedge f(x_1, ..., x_n)(y)
\end{array}
\]

Then:
\begin{itemize}

\item $g$ is an extension of $f$ on the first $k$ arguments.

\item For all $h$ extension of $f$ on the first $k$ arguments
we have $h \approx g$

\end{itemize}
\end{theorem}

\end{toappendix}

\begin{toappendix}
[
\begin{proof}

For readability we are only going to prove the case where $n = k = 1$ and
$X = X_1$.
This proof can easily be adapted to the case with several arguments and where
the extension is not necessarily on all the arguments.
See proof in Appendix \thisappendix ~ for the general case which works for any
values of $k$ and $n$.
\begin{itemize}

\item First we prove that $g : El(X) \rightarrow El(Y)$ which means
proving stability of $g$.
If $\eqeffEl X u {u'}$ : We want to prove that
$\eqeffEl Y {g(u)} {g(u')}$.
\begin{itemize}

\item If $g(u)(y)$ and $|y =_Y y'|$ then by definition of $g$,
there exists $x \in |X|$, such that $u(x)$ and $f(x)(y)$.
Then because $\eqeffEl X u {u'}$ we have
$\eqeff X x x$ by using unicity of $u$.
Hence, $\eqeffEl Y {f(x)} {f(x)}$ by stability of $f$.
So, $f(x)(y')$ by stability of $f(x)$.
Hence, we have $g(u)(y')$.

\item If $g(u)(y)$ and $g(u)(y')$ then by definition of $g$,
there exists $x$ and $x'$ in $|X|$ such that $u(x)$, $u(x')$, $f(x)(y)$ and
$f(x')(y')$.
Because $\eqeffEl X u {u'}$ we have $\eqeff X x {x'}$ by unicity of $u$.
Hence $\eqeffEl Y {f(x)} {f(x')}$ by stability of $f$.
So we have $f(x)(y')$ by equivalence between $f(x)$ and $f(x')$.
And then $|y =_Y y'|$ by unicity of $f(x)$.

\item $\eqeffEl X u {u'}$, so there exists $x \in |X|$ such that $u(x)$.
Then we have $|x =_X x|$ by unicity of $u$.
Hence $\eqeffEl Y {f(x)} {f(x)}$ by stability of $f$.
So there exists $y \in |Y|$, such that $f(x)(y)$.
Then we have $g(u)(y)$ by definition of $g$.
Hence there exists $y \in |Y|$ such that $g(u)(y)$.

\item If $g(u)(y)$ then there exists $x \in |X|$, such that $u(x)$ and
$f(x)(y)$. So we have $u'(x)$. Hence we have $g(u')(y)$.

\item By a similar argument, if $g(u')(y)$ then $g(u)(y)$.

\end{itemize}
Hence $\eqeffEl Y {g(u)} {g(u')}$.
Therefore $g : El(X) \rightarrow El(Y)$.

\item If $|x =_X x|$:
\begin{itemize}

\item Then $\eqeffEl Y {f(x)} {f(x)}$, $\eqeffEl X {el_X(x)} {el_X(x)}$
(from Lemma~\ref{lem:LowLevelHighLevel}.4) and
$\eqeffEl Y {g(el_X(x))} {g(el_X(x))}$

\item If $g(el_X(x))(y)$ then there exists $x' \in |X|$ such that $el_X(x)(x')$
and $f(x')(y)$. Hence $\eqeff X x {x'}$.
So we have $\eqeffEl Y {f(x)} {f(x')}$.
Hence we have $f(x)(y)$.

\end{itemize}
Hence we have $\eqeffEl Y {g(el_X(x))} {f(x)}$ by Lemma
\ref{lem:LowLevelHighLevel}.2.
Therefore,  $g$ is an extension of $f$.

\item If $h$ is an extension of $f$.
We construct $F \in |El(X)| \rightarrow Prop$ defined by
$F(u) \eqdef \eqeffEl Y {g(u)} {h(u)}$.
By stability of $g$ and $h$ we have $F \in SP(X)$.

If $\eqeff X x x$, then, because $g$ and $h$ are extensions of $f$ we have
$\eqeffEl Y {g(el_X(x))} {f(x)}$ and $\eqeffEl Y {h(el_X(x))} {f(x)}$.
Hence we have
$\eqeffEl Y {g(el_X(x))} {h(el_X(x))}$ which is $F(el_X(x))$.

By Theorem \ref{th:ExtensionTruth}, we can prove that
$\vDash \forall u \in El(X), F(u)$.

Therefore, $g \approx h$ by definition of $F$ and $\approx$.

\end{itemize}
\end{proof}
]

\begin{proof}

\begin{itemize}

\item If $\eqeffEl {X_1} {u_1} {u_1'}$, \ldots $\eqeffEl {X_k} {u_k} {u_k'}$,
$\eqeff {X_{k + 1}} {x_{k + 1}} {x_{k + 1}'}$, \ldots
$\eqeff {X_n} {x_n} {x_n'}$:
we want to prove that
$\eqeffEl Y {g(u_1, \ldots u_k, x_{k + 1}, \ldots x_n)}
{g(u_1', \ldots u_k', x_{k + 1}', \ldots x_n')}$.

\begin{itemize}

\item If $g(u_1, \ldots u_k, x_{k + 1}, \ldots x_n)(y)$ and
$\eqeff Y y {y'}$ then there exist $x_1 \in |X_1|$, \ldots $x_k \in |X_k|$,
such that $u_1(x_1)$, \ldots $u_k(x_k)$ and $f(x_1, \ldots x_n)(y)$.
Then we have $\eqeff {X_1} {x_1} {x_1}$, \ldots $\eqeff {X_k} {x_k} {x_k}$.
Moreover we have $\eqeff {X_{k + 1}} {x_{k + 1}} {x_{k + 1}}$, \ldots
$\eqeff {X_n} {x_n} {x_n}$.
Hence $\eqeffEl Y {f(x_1, \ldots x_n)} {f(x_1, \ldots x_n)}$.
So $f(x_1, \ldots x_n)(y')$.
Hence we have \\ $g(u_1, \ldots u_k, x_{k + 1}, \ldots x_n)(y)$

\item If $g(u_1, \ldots u_k, x_{k + 1}, \ldots x_n)(y)$ and
$g(u_1, \ldots u_k, x_{k + 1}, \ldots x_n)(y')$ then there exist
$x_1, x_1' \in |X_1|$, \ldots $x_k, x_k' \in |X_k|$ such that
$u_1(x_1)$, \ldots $u_k(x_k)$, $u_1(x_1')$, \ldots $u_k(x_k')$,
$f(x_1, \ldots x_n)(y)$ and $f(x_1', \ldots x_k', x_{k + 1}, \ldots x_n)(y')$.
Therefore we have $\eqeff {X_1} {x_1} {x_1'}$, \ldots
$\eqeff {X_k} {x_k} {x_k'}$.
We also have  $\eqeff {X_{k + 1}} {x_{k + 1}} {x_{k + 1}}$, \ldots
$\eqeff {X_n} {x_n} {x_n}$.
Hence $\eqeffEl Y {f(x_1, \ldots x_n)}
{f(x_1', \ldots x_k', x_{k + 1}, \ldots x_n)}$.
So we have $f(x_1, \ldots x_n)(y')$.
And then $\eqeff Y y {y'}$.

\item $\eqeffEl {X_1} {u_1} {u_1}$, \ldots $\eqeffEl {X_k} {u_k} {u_k}$,
so there exist $x_1 \in |X_1|$, \ldots $x_k \in |X_k|$, such that
$u_1(x_1)$, \ldots $u_k(x_k)$.
Then we have $\eqeff {X_1} {x_1} {x_1}$, \ldots
$\eqeff {X_k} {x_k} {x_k}$.
Hence $\eqeffEl Y {f(x_1, \ldots x_n)}
{f(x_1, \ldots x_k, x_{k + 1}', \ldots x_n')}$.
So there exists $y \in |Y|$, such that $f(x_1, \ldots x_n)(y)$.
Then we have $g(u_1, \ldots, u_k, x_{k + 1}, \ldots x_n)(y)$.
Hence there exists $y \in |Y|$ such that \\
$g(u_1, \ldots u_n, x_{k + 1}, \ldots x_n)(y)$.

\item If $g(u_1, \ldots u_k, x_{k + 1}, \ldots x_n)(y)$ then there exist
$x_1 \in |X_1|$, \ldots $x_k \in |X_x|$ such that
$u_1(x_1)$, \ldots $u_k(x_k)$ and $f(x_1, \ldots x_n)(y)$.
So we have $u_1'(x_1)$, \ldots $u_k'(x_k)$.
Moreover, $\eqeff {X_1} {x_1} {x_1}$, \ldots $\eqeff {X_k} {x_k} {x_k}$.
Therefore $\eqeffEl Y {f(x_1, \ldots x_n)}
{f(x_1, \ldots x_k, x_{k + 1}', \ldots x_n')}$.
Hence, $f(x_1, \ldots x_k, x_{k + 1}', \ldots x_n')(y)$.
So we have $g(u_1', \ldots u_k', x_{k + 1}', \ldots x_n')(y)$.

\item With a similar argument,
if $g(u_1', \ldots u_k', x_{k + 1}', \ldots x_n')(y)$ then
$g(u_1, \ldots u_k, x_{k + 1}, \ldots x_n)(y)$.

\end{itemize}

Hence $\eqeffEl Y {g(u_1, \ldots u_k, x_{k + 1} \ldots x_n)}
{g(u_1', \ldots u_k', x_{k + 1}', \ldots x_n')}$.
Therefore $g : El(X_1) \rightarrow \ldots El(X_k) \rightarrow \ldots
X_{k + 1} \rightarrow \ldots X_n \rightarrow El(Y)$.

\item If $\eqeff {X_1} {x_1} {x_1}$, \ldots $\eqeff {X_n} {x_n} {x_n}$:

\begin{itemize}

\item Then $\eqeffEl Y {f(x_1, \ldots x_n)} {f(x_1, \ldots x_n)}$,
$\eqeffEl {X_1} {el_{X_1}(x_1)} {el_{X_1}(x_1)}$, \ldots
$\eqeffEl {X_k} {el_{X_k}(x_k)} {el_{X_k}(x_k)}$,
and
\[\eqeffEl Y {g(el_{X_1}(x_1), \ldots el_{X_k}(x_k), x_{k + 1}, \ldots x_n)} 
{g(el_{X_1}(x_1), \ldots el_{X_k}(x_k), x_{k + 1}, \ldots x_n)}\].

\item If  $g(el_{X_1}(x_1), \ldots el_{X_k}(x_k), x_{k + 1}, \ldots x_n)(y)$
then there exist $x_1' \in |X_1|$, \ldots $x_k' \in |X_k|$ such that
$el_{X_1}(x_1)(x_1')$, \ldots $el_{X_k}(x_k)(x_k')$ and
$f(x_1', \ldots x_k', x_{k + 1}, \ldots x_n)(y)$.
Hence $\eqeff {X_1} {x_1} {x_1'}$, \ldots $\eqeff {X_k} {x_k} {x_k'}$.
So we have $\eqeffEl Y {f(x_1, \ldots x_n)}
{f(x_1', \ldots x_k', x_{k + 1}, \ldots x_n)}$.
Hence we have $f(x_1, \ldots x_n)(y)$.

\end{itemize}

Hence we have $\eqeffEl Y
{g(el_{X_1}(x_1), \ldots el_{X_k}(x_k), x_{k + 1}, \ldots x_n)}
{f(x_1, \ldots x_n)}$.
Therefore, $g$ is an extension of $f$ on the first $k$ arguments.

\item If $h$ is an extension of $f$ on the first $k$ arguments.
We construct $F \in (|El(X_1)| \times \ldots |El(X_k)| \times
|X_{k + 1}| \times \ldots |X_n|) \rightarrow Prop$ defined by:
\[F(u_1, \ldots u_k, x_{k + 1}, \ldots x_n) \eqdef
\eqeff Y {g(u_1, \ldots u_k, x_{k + 1}, \ldots x_n)}
{h(u_1, \ldots u_k, x_{k + 1}, \ldots x_n)}\].
By stability of $g$ and $h$ we have
$F \in SP(El(X_1), \ldots El(X_k), X_{k + 1}, \ldots X_n)$.

If $\eqeff {X_1} {x_1} {x_1}$, \ldots $\eqeff {X_n} {x_n} {x_n}$ then
we have $\eqeffEl Y
{g(el_{X_1}(x_1), \ldots el_{X_k}(x_k), x_{k + 1}, \ldots x_n)}
{f(x_1, \ldots x_n)}$ and
$\eqeffEl Y
{h(el_{X_1}(x_1), \ldots el_{X_k}(x_k), x_{k + 1}, \ldots x_n)}
{f(x_1, \ldots x_n)}$.
Hence $F(el_{X_1}(x_1), \ldots el_{X_k}(x_k), x_{k + 1}, \ldots x_n)$.

By Theorem \ref{th:ExtensionTruth}, we can prove that:
\[\vDash \forall u_1 \in El(X_1), \ldots u_k \in El(X_k),
x_{k + 1} \in X_{k + 1}, \ldots x_n \in X_n,
F(u_1, \ldots u_n, x_{k + 1}, \ldots x_n)\]

Therefore $g \approx h$.

\end{itemize}

\end{proof}

\end{toappendix}


\section{The Effective Topos}

In this section we construct a category that we call \emph{The effective Topos}.
To prove that this category is a topos with an object of
natural integers, we prove results in the high level part of our
framework.
Except for the axiom of unique choice, these results are extensions from the low level to the high level.
We also prove that we can do program extraction in the framework.
The fact that our definition of the Effective Topos is equivalent Hyland's is
given in Appendix \ref{sec:Usual}.
We say that we \emph{extend} the framework if we define new operators and prove
new properties in this framework.
In section \ref{sec:Cat} we construct the Effective Topos.
In section \ref{sec:CCC} we extend the framework to prove that this category
is Cartesian closed.
In section \ref{sec:SubObjectClassifier} we extend the framework to prove that
this category has a sub-object classifier (therefore it is a topos).
In section \ref{sec:Nat} we extend the framework to prove that this category
has an object of natural integers and we also prove that we can do program
extraction with our framework.

\subsection{Definition of the category}

\label{sec:Cat}

To build the category we must first check that the composition is well-defined
for functions of this framework.

\begin{toappendix}

  \appendixbeyond 0

  \begin{lemma}[Correctness of composition]

    \label{lem:SoundComp}

Assume $X$, $Y$ and $Z$ are effective sets:
    \begin{itemize}

    \item If $f : X \rightarrow Y$ and $g : Y \rightarrow Z$, then
      $g \circ f : X \rightarrow Z$.

    \item If $f : X \rightarrow Y$, $f' : X \rightarrow Y$, $g : Y \rightarrow Z$,
      $g' : Y \rightarrow Z$, $f \approx f'$ and $g \approx g'$, then
      $g \circ f \approx g' \circ f'$.

    \end{itemize}
  \end{lemma}

\end{toappendix}

\begin{toappendix}[
  \begin{proof}
    Straightforward. See Appendix \thisappendix.
  \end{proof}]

  \begin{proof}


    \begin{itemize}

    \item If $\eqeffEl X u {u'}$ then $\eqeffEl Y {f(u)} {f(u')}$ because
      $f : X \rightarrow Y$.
      Hence $\eqeffEl Z {g(f(u))} {g(f(u'))}$ because $g : Y \rightarrow Z$.
      Therefore $g \circ f : X \rightarrow Z$.

    \item If $\eqeffEl X u {u}$ then $\eqeffEl Y {f(u)} {f'(u)}$ because
      $f \approx f'$.
      So $\eqeffEl Y {f(u)} {f(u)}$ by symmetry and transitivity.
      Therefore $\eqeffEl Z {g(f(u))} {g'(f(u))}$ because $g \approx g'$.
      We also have
      $\eqeffEl Z {g'(f(u))} {g'(f'(u))}$ because $g' : Y \rightarrow Z$.
      Hence, by transitivity, we have $\eqeffEl Z {g(f(u))} {g'(f'(u))}$.
      Therefore $g \circ f \approx g' \circ f'$.

    \end{itemize}

  \end{proof}

\end{toappendix}
\begin{theorem}[Definition of the category \textbf{C}]

  We build \textbf{C} the category whose objects are the effective sets and the
  morphisms between two objects $X$ and $Y$ are the functions \\
  $f : El(X) \rightarrow El(Y)$
  modulo the $\approx$ relation. And the composition is the usual composition.

  It is straightforward that \textbf{C} is a category
  (the identity is the usual identity).

\end{theorem}
We call this category the effective topos.
We will prove (in Appendix \ref{sec:Usual}) that this definition is equivalent
to the usual definition by Hyland.

\subsection{This category is Cartesian closed}

A Cartesian closed category is a category that has a final object, products,
and power objects (closure over the product).

\label{sec:CCC}

\subsubsection{Final object}

\label{sec:Final}

Let $S$ be a set and $()$ an element of $S$.
\begin{definition}[Definition of $1$]

We build $1$ the effective set defined by:
\begin{itemize}

\item $|1| \eqdef S$

\item $|x =_1 y| \eqdef \top$

\end{itemize}
It is straightforward that $1$ is an effective set.

\end{definition}

Here is the high level constant (defined with the low level one):

\begin{definition}[Definition of $<>$]

We define $<> \in |El(1)|$ by $<> \eqdef el_1(())$.

\end{definition}

Here are the high level properties we need to prove that \textbf{C} has a
final object.

\begin{toappendix}

\appendixbeyond 0

\begin{theorem}[Properties of 1]\strut
\label{th:FinalProperties}
\begin{itemize}

\item $\vDash | <> =_1 <> |$

\item $\vDash \forall u, v \in El(1), |u =_1 v|$

\end{itemize}
\end{theorem}

\end{toappendix}

\begin{toappendix}
[
\begin{proof}

Straightforward.
To prove the second item, we use Theorem \ref{th:ExtensionTruth}.
See Appendix \thisappendix.

\end{proof}
]

\begin{proof}

\begin{itemize}

\item We have $\eqeff 1 {()} {()}$ by definition of $1$.
By Lemma $\ref{lem:LowLevelHighLevel}$.4 we have
$\eqeffEl 1 {el_1(())} {el_1(())}$.
Therefore $\eqeffEl 1 {<>} {<>}$.

\item We construct $F \in (|El(1)| \times |El(1)|) \rightarrow Prop$
defined by $F(u, v) \eqdef \eqeffEl 1 u v$.
It is trivial that $F \in SP(El(1), El(1))$.

If $\eqeff 1 x x$ and $\eqeff 1 y y$:
By definition of $1$, we have $\eqeff 1 x y$.
Then  by Lemma \ref{lem:LowLevelHighLevel}.4 we have
$\eqeffEl 1 {el_1(x)} {el_1(y)}$.
Therefore $F(el_1(x), el_1(y))$.

By Theorem \ref{th:ExtensionTruth} we have
$\vDash \forall u, v \in El(1), F(u, v)$.
Then we can conclude.

\end{itemize}

\end{proof}

\end{toappendix}

Now with the previous theorem we can easily prove the categorical result:

\begin{toappendix}

\appendixbeyond 0

\begin{theorem}[\textbf{C} has a final object]

\label{th:Final}

$1$ is the final object of \textbf{C}.

\end{theorem}

\end{toappendix}

\begin{toappendix}
[
\begin{proof}
Corollary of Theorem \ref{th:FinalProperties}.
See Appendix \thisappendix.
\end{proof}
]

\begin{proof}

Assume $X$ is an effective set.

We construct $f \in |El(X)| \rightarrow |El(1)|$ defined by:
$f(u) \eqdef <>$.

\begin{itemize}

\item Stability of $f$: If $\eqeffEl X u {u'}$ then by
Theorem~\ref{th:FinalProperties} we have
$\eqeffEl 1 {<>} {<>}$.
Hence $\eqeffEl 1 {f(u)} {f(u')}$.
Therefore $f : El(X) \rightarrow El(1)$.

\item Unicity of $f$: Assume $g : El(X) \rightarrow El(1)$.
If $\eqeffEl X u u$, then, because $f, g : El(X) \rightarrow El(1)$,
we have $\eqeffEl 1 {f(u)} {f(u)}$ and $\eqeffEl 1 {g(u)} {g(u)}$.
By Theorem \ref{th:FinalProperties} we have $\eqeffEl 1 {f(u)} {g(u)}$.
Therefore $f \approx g$.

\end{itemize}

Hence, $1$ is a final object of \textbf{C}.

\end{proof}

\end{toappendix}

\subsubsection{Product}

To prove that \textbf{C} has products we use the same method as in
Section~\ref{sec:Final}.

This is the minimum requirement to prove that \textbf{C} has products:

If $X$ and $Y$ sets, then $X \times Y$ is a set. If $z \in X \times Y$, then
$\pi_1(z) \in X$ and $\pi_2(z) \in Y$. And if $x \in X$ and $y \in Y$, then
$(x, y) \in X \times Y$, $\pi_1(x, y) = x$ and $\pi_2(x, y) = y$.

The definition of the product is the same as the one in the usual presentation
of the effective topos.

\begin{definition}[Definition of $A \times B$]

If $A$ and $B$ effective sets, then we define $A \times B$ by:
\begin{itemize}

\item $|A \times B| \eqdef |A| \times |B|$

\item $\eqeff {A \times B} c {c'} \eqdef \eqeff A {\pi_1(c)} {\pi_1(c')} \wedge
\eqeff B {\pi_2(c)} {\pi_2(c')}$

\end{itemize}
It is trivial to show that $A \times B$ is an effective set.

\end{definition}

For the high level operators it will be less trivial than in the case of the
final object: we define the high level constructions by using Theorem
\ref{th:ExistenceUnicityExtensionFun}.

\begin{definition}[Definition of $p_1'$, $p_2'$ and $cons'$]

We construct $p_1' \in |A \times B| \rightarrow |El(A)|$,\linebreak
$p_2' \in |A \times B| \rightarrow |El(B)|$ and
$cons' \in (|A| \times |B|) \rightarrow |El(A \times B)|$ as follows:

$p_1'(c) \eqdef el_A(\pi_1(c))$,
$p_2'(c) \eqdef el_B(\pi_2(c))$,
$cons'(a, b) \eqdef el_{A \times B}((a, b))$

\end{definition}

\begin{toappendix}

\appendixbeyond 0

\begin{lemma}[Stability of $p_1'$, $p_2'$ and $cons'$]\strut

\label{lem:ProductPreStable}

$p_1' : (A \times B) \rightarrow El(A)$,
$p_2' : (A \times B) \rightarrow El(B)$ and
$cons' : A \rightarrow B \rightarrow El(A \times B)$.

\end{lemma}

\end{toappendix}

\begin{toappendix}
[
\begin{proof}
Straightforward. See Appendix A.
\end{proof}
]

\begin{proof}

\begin{itemize}

\item If $\eqeff {A \times B} c {c'}$ then by definition of $A \times B$ we have
$\eqeff A {\pi_1(c)} {\pi_1(c')}$.
So $\eqeffEl A {el_A(\pi_1(c))} {el_A(\pi_1(c'))}$.
Hence $\eqeffEl A {p_1'(c)} {p_1'(c')}$.
Therefore $p_1' : (A \times B) \rightarrow El(A)$.

\item With the same kind of proof, we have
$p_2' : (A \times B) \rightarrow El(B)$.

\item If $\eqeff A a {a'}$ and $\eqeff B b {b'}$ then
$\eqeff A {\pi_1(a, b)} {\pi_1(a', b')}$ and
$\eqeff B {\pi_2(a, b)} {\pi_2(a', b')}$.
By definition of $A \times B$ we have $\eqeff {A \times B} {(a, b)} {(a', b')}$.
So,
$\eqeffEl {A \times B} {el_{A \times B}(a, b)} {el_{A \times B}(a', b')}$.
Hence $\eqeffEl {A \times B} {cons'(a, b)} {cons'(a', b')}$.
Therefore $cons' : A \rightarrow B \rightarrow El(A \times B)$.

\end{itemize}

\end{proof}

\end{toappendix}

\begin{definition}[Definition of $p_1$, $p_2$ and $<u, v>$]

Assume $A$ and $B$ are effective sets.
\begin{itemize}

\item We write $p_1 : El(A \times B) \rightarrow El(A)$ the extension of $p_1'$.

\item We write $p_2 : El(A \times B) \rightarrow El(B)$ the extension of $p_2'$.

\item We write $cons : El(A) \rightarrow El(B) \rightarrow El(A \times B)$
the extension of $cons'$.\\
We write $<u, v>$ for $cons(u, v)$.
\end{itemize}
\end{definition}

The high level constructions satisfy stability, $\beta$-equivalence,
and extensionality:

\begin{toappendix}

\appendixbeyond 0

\begin{theorem}[Properties of $A \times B$]\strut

\label{th:ProductProperties}

Assume $A$ and $B$ are effective sets.
\begin{itemize}

\item $\vDash \forall w, w' \in |El(A \times B)|,
\eqeffEl {A \times B} w {w'} \Rightarrow \eqeffEl A {p_1(w)} {p_1(w')}$

\item $\vDash \forall w, w' \in |El(A \times B)|,
\eqeffEl {A \times B} w {w'} \Rightarrow \eqeffEl B {p_2(w)} {p_2(w')}$

\item $\vDash \forall u, u' \in |El(A)|, v, v' \in |El(B)|,\\
\strut\qquad\qquad
\eqeffEl A u {u'} \Rightarrow \eqeffEl B v {v'} \Rightarrow 
\eqeffEl {A \times B} {<u, v>} {<u', v'>}$

\item $\vDash \forall u \in El(A), v \in El(B), \eqeffEl A {p_1(<u, v>)} u$

\item $\vDash \forall u \in El(A), v \in El(B), \eqeffEl B {p_2(<u, v>)} v$

\item $\vDash \forall w, w' \in El(A \times B),
\eqeffEl A {p_1(w)} {p_1(w')} \Rightarrow \eqeffEl B {p_2(w)} {p_2(w')}
\Rightarrow \eqeffEl {A \times B} w {w'}$

\end{itemize}
\end{theorem}

\end{toappendix}

\begin{toappendix}
[
\begin{proof}
Straightforward with the use Theorem \ref{th:ExtensionTruth} and the definitions
of $p_1$, $p_2$ and $cons$.
See Appendix \thisappendix.
\end{proof}
]

\begin{proof}
\begin{itemize}

\item The first three properties are just expressing the stability of
$p_1$, $p_2$ and $cons$ which theu are by construction.

\item We use Theorem~\ref{th:ExtensionTruth}:
We construct $F \in (|El(A)| \times |El(B)|) \rightarrow Prop$ defined by:
$F(u, v) \eqdef \eqeffEl A {p_1(cons(u, v))} u$.
By stability of $p_1$ and $cons$ we have $F \in SP(El(A), El(B))$.

If $\eqeff A a a$ and $\eqeff B b b$ then by definition of $cons$ we have
$\eqeffEl {A \times B} {cons(el_A(a), el_B(b))} {el_{A \times B}((a, b))}$.
Hence, $\eqeffEl A {p_1(cons(el_A(a), el_B(b)))}
{p_1(el_{A \times B}((a, b)))}$.
Moreover, we have $\eqeff {A \times B} {(a, b)} {(a, b)}$
(because $\eqeff A {\pi_1(a, b)} {\pi_1(a, b)}$ and
$\eqeff B {\pi_2(a, b)} {\pi_2(a, b)}$).
By definition of $p_1$ we have
$\eqeffEl {A \times B} {p_1(el_{A \times B}(a, b))} {el_A(\pi_1(a, b))}$ with
$\pi_1(a, b) = a$.
Hence, by transitivity, we have
$\eqeffEl A {p_1(cons(el_A(a), el_B(b)))} {el_A(a)}$.
Therefore, $F(el_A(a), el_B(b))$.
By Theorem \ref{th:ExtensionTruth} we have
$\vDash \forall u \in El(A), v \in El(B), F(u, v)$.
Then we can conclude.

\item With the same kind of proof as above we can prove the other
$\beta$-reduction (with $p_2$).

\item We use Theorem~\ref{th:ExtensionTruth}: We construct
$F \in (|El(A \times B)| \times |El(A \times B)|) \rightarrow
Prop$ defined by:
\[F(w, w') \eqdef \eqeffEl A {p_1(w)} {p_1(w')} \Rightarrow
\eqeffEl B {p_2(w)} {p_2(w')} \Rightarrow \eqeffEl {A \times B} w {w'}\]
By stability of $p_1$ and $p_2$, we have
$F \in SP(El(A \times B), El(A \times B))$.

If $\eqeff {A \times B} c c$, $\eqeff {A \times B} {c'} {c'}$,
$\eqeffEl A {p_1(el_{A \times B}(c))} {p_1(el_{A \times B}(c'))}$ and
$\eqeffEl B {p_2(el_{A \times B}(c))} {p_2(el_{A \times B}(c'))}$:
Then by definition of $p_1$,
$\eqeffEl A {p_1(el_{A \times B}(c))} {el_A(\pi_1(c))}$ and \\
$\eqeffEl A {p_1(el_{A \times B}(c'))} {el_A(\pi_1(c'))}$.
By symmetry and transitivity, $\eqeffEl A {el_A(\pi_1(c))} {el_A(\pi_1(c'))}$.
Hence $\eqeff A {\pi_1(c)} {\pi_1(c')}$.
We can also prove that $\eqeff B {\pi_2(c)} {\pi_2(c')}$.
By definition of $A \times B$, $\eqeff {A \times B} c {c'}$.
And then $\eqeffEl {A \times B} {el_{A \times B}(c)} {el_{A \times B}(c')}$.

By Theorem \ref{th:ExtensionTruth} we have
$\vDash \forall w, w' \in El(A \times B), F(w, w')$.
Then we can conclude.

\end{itemize}

\end{proof}

\end{toappendix}

With this theorem we can conclude the categorical result:

\begin{toappendix}

\appendixbeyond 0

\begin{theorem}[\textbf{C} has products]\strut

\label{th:Product}

If $A$ and $B$ effective sets, then $A \times B$ is the product of $A$ and
$B$ in \textbf{C}, $p_1$ and $p_2$, its projections.

Hence \textbf{C} is cartesian.

\end{theorem}

\end{toappendix}

\begin{toappendix}
[
\begin{proof}
Corollary of Theorem \ref{th:ProductProperties}.
See Appendix \thisappendix.
\end{proof}
]

\begin{proof}

$A \times B$ is an effective set,
$p_1 : El(A \times B) \rightarrow El(A)$ and
$p_2 : El(A \times B) \rightarrow El(B)$.

Assume $X$ is an effective set, $f : El(X) \rightarrow El(A)$ and
$g : El(X) \rightarrow El(B)$.

We construct $\varphi \in |El(X)| \rightarrow |El(A \times B)|$
defined by: $\varphi(u) \eqdef <f(u), g(u)>$.

\begin{itemize}

\item Stability of $\varphi$: If $\eqeffEl X u {u'}$ then
$\eqeff A {f(u)} {f(u')}$ by stability of $f$.
By stability of $g$ we also have $\eqeff B {g(u)} {g(u')}$.
So, by stability of $cons$, $\eqeff {A \times B} {<f(u), g(u)>}
{<f(u'), g(u')>}$.
Hence, $\eqeff {A \times B} {\varphi(u)} {\varphi(u')}$.
Therefore $\varphi : El(X) \rightarrow El(A \times B)$.

\item If $\eqeffEl X u u$ then $\eqeffEl A {f(u)} {f(u)}$ and
$\eqeffEl B {g(u)} {g(u)}$.
By Theorem \ref{th:ProductProperties} we have
$\eqeffEl A {p_1(<f(u), g(u)>)} {f(u)}$.
Hence $\eqeffEl A {(p_1 \circ \varphi)(u)} {f(u)}$.
Therefore, $p_1 \circ \varphi \approx f$.

\item By a similar argument we can also prove that
$p_2 \circ \varphi \approx g$.

\item Assume $\psi : El(X) \rightarrow El(A \times B)$ such that
$p_1 \circ \psi \approx f$ and $p_2 \circ \psi \approx g$.

If $\eqeffEl X u u$ then $\eqeffEl {A \times B} {\varphi(u)} {\varphi(u)}$,
$\eqeffEl {A \times B} {\psi(u)} {\psi(u)}$,
$\eqeffEl A {p_1(\varphi(u))} {f(u)}$ and $\eqeffEl A {p_1(\psi(u))}
{f(u)}$.
So, $\eqeffEl A {p_1(\varphi(u))} {p_1(\psi(u))}$.
We can also prove that $\eqeffEl B {p_2(\varphi(u))} {p_2(\psi(u))}$.
Hence, by Theorem \ref{th:ProductProperties},
$\eqeffEl {A \times B} {\varphi(u)} {\psi(u)}$.
Therefore $\varphi \approx \psi$.

\end{itemize}

\end{proof}

\end{toappendix}

\subsubsection{Closure}

Compared to Heyting's, because we do not have the same notion of morphisms,
the definition of the power object is not the same: it is much simpler.

\begin{definition}[Definition of $A \Rrightarrow B$]\strut

Assume $A$ and $B$ are effective sets.
Then we define the effective set $A \Rrightarrow B$ as follows:
\begin{itemize}

\item $|A \Rrightarrow B| \eqdef |El(A)| \rightarrow |El(B)|$

\item $\eqeff {A \Rrightarrow B} f g \eqdef
(\forall u, u' \in |El(A)|, \eqeffEl A u {u'} \Rightarrow
\eqeffEl B {f(u)} {g(u')})$

\end{itemize}
It is straightforward to prove that $A \Rrightarrow B$ is an effective set.

$\eqeff {A \Rrightarrow B} f g$ is equivalent to expressing internaly that:
$f$ is stable, $g$ is stable and $f$ and $g$ are equivalent.
\end{definition}

As in the product case, we define high level constructors and destructors.

\begin{definition}[Definition of $app'$]\strut

We construct $app' \in (|A \Rrightarrow B| \times |El(A)|) \rightarrow |El(B)|$
defined by: $app'(f, u) \eqdef f(u)$

\end{definition}

\begin{toappendix}

\appendixbeyond 0

\begin{lemma}[Stability of $app'$]

\label{lem:PreClosure}

$app' : (A \Rrightarrow B) \rightarrow El(A) \rightarrow El(B)$

\end{lemma}

\end{toappendix}

\begin{toappendix}
[
\begin{proof}
Straightforward. See Appendix \thisappendix.
\end{proof}
]

\begin{proof}

If $\eqeff {A \Rrightarrow B} f {f'}$ and $\eqeffEl A u {u'}$ then by definition
of $A \Rrightarrow B$ we have $\eqeffEl B {f(u)} {f'(u)}$.
Hence $\eqeffEl B {app'(f, u)} {app'(f', u')}$.
Therefore $app' : (A \Rrightarrow B) \rightarrow El(A) \rightarrow El(B)$.

\end{proof}

\end{toappendix}

\begin{definition}[Definition of $app$ and $\Lambda u : A . f(u)$]

Assume $A$ and $B$ are effective sets.
\begin{itemize}

\item We write $app : El(A \Rrightarrow B) \rightarrow El(A) \rightarrow El(B)$ 
the extension of $app'$ on the first argument.

\item If $f \in |El(A)| \rightarrow |El(B)|$ we write
$\Lambda u : A . f(u)$ for $el_{A \Rrightarrow B}(u \in |El(A)| \mapsto f(u))$.

\end{itemize}
\end{definition}

And as in the product case, we can prove stability, $\beta$-equivalence and
extensionality on these high level operators.

\begin{toappendix}

\appendixbeyond 0

\begin{theorem}[Properties of $A \Rrightarrow B$]

\label{th:ClosureProperties}

Assume $A$ and $B$ are effective sets.
\begin{itemize}

\item $\vDash \forall w, w' \in |El(A \Rrightarrow B)|, u, u' \in |El(A)|,\\
\strut\qquad\qquad\eqeffEl A w {w'} \Rightarrow \eqeffEl A u {u'} \Rightarrow
\eqeffEl B {app(w, u)} {app(w', u')}$


\item $\vDash \forall f, f' \in |El(A)| \rightarrow |El(B)|,
(\forall u, u' \in |El(A)|,\\
\strut\qquad\qquad \eqeffEl A u {u'} \Rightarrow
\eqeffEl B {f(u)} {f'(u')}) \Rightarrow
\eqeffEl {A \Rrightarrow B} {\Lambda u : A . f(u)}
{\Lambda u : A . f'(u)}$

\item $\vDash \forall f \in |El(A)| \rightarrow |El(B)|, u \in El(A), \\
\strut\qquad\qquad
(\forall v, v' \in |El(A)|, \eqeffEl A v {v'} \Rightarrow
\eqeffEl B {f(v)} {f(v')})\\
\strut\qquad\qquad \Rightarrow
\eqeffEl B {app((\Lambda v : A . f(v)), u)} {f(u)}$

\item $\vDash \forall w, w' \in El(A \Rrightarrow B),
(\forall u \in El(A), \eqeffEl B {app(w, u)} {app(w', u)}) \Rightarrow
\eqeffEl {A \Rrightarrow B} w {w'}$

\end{itemize}
\end{theorem}

\end{toappendix}

\begin{toappendix}
[
\begin{proof}
Only the last point is not trivial and it can be proven by using
Theorem \ref{th:ExtensionTruth}.
See Appendix \thisappendix.
\end{proof}
]

\begin{proof}

\begin{itemize}

\item The first point is just expressing the stability of $app$.

\item The second point is just an way of expressing that if
$\eqeff {A \Rrightarrow B} f {f'}$ then $\eqeffEl {A \Rrightarrow B}
{el_{A \Rrightarrow B}(f)} {el_{A \Rrightarrow B}(f')}$.

\item The third point is just another way of expressing that $app$ is
an extension of $app'$ on the first argument.

\item We use Theorem~\ref{th:ExtensionTruth}: We construct
$F \in (|El(A \Rrightarrow B)| \times |El(A \Rrightarrow B)|) \rightarrow Prop$
defined by:
\[F(w, w') \eqdef (\forall u \in El(A), \eqeffEl B {app(w, u)} {app(w', u)})
\Rightarrow \eqeffEl {A \Rrightarrow B} w {w'}\]
By stability of $app$, $F \in SP(El(A \Rrightarrow B), El(A \Rrightarrow B))$.

If $\eqeff {A \Rrightarrow B} f f$, $\eqeff {A \Rrightarrow B} {f'} {f'}$ and
$\forall u \in El(A), \eqeffEl B {app(el_{A \Rrightarrow B}(f), u)}
{app(el_{A \Rightarrow B}(f'), u)}$:
Let $u, u' \in |El(A)|$ such that $\eqeffEl A u {u'}$.
So $\eqeffEl A u u$.
By definition of $app$,
$\eqeffEl B {app(el_{A \Rrightarrow B}(f), u)} {f(u)}$.
By hypothesis, $\eqeffEl B {app(el_{A \Rrightarrow B}(f), u)}
{app(el_{A \Rightarrow B}(f'), u)}$.
We can also prove that
$\eqeffEl B {app(el_{A \Rrightarrow B}(f'), u')}
{f'(u')}$.
Hence $\eqeffEl B {f(u)} {f(u')}$.
Therefore $\eqeff {A \Rrightarrow B} f {f'}$ and
$\eqeffEl {A \Rrightarrow B} {el_{A \Rrightarrow B}(f)}
{el_{A \Rrightarrow B}(f')}$.

By Theorem \ref{th:ExtensionTruth} we have
$\vDash \forall w, w' \in El(A \Rrightarrow B), F(w, w')$.
Then we can conclude.

\end{itemize}

\end{proof}

\end{toappendix}

And then we can prove this theorem as easy as in the case of the category of
sets.

\begin{toappendix}

\appendixbeyond 0

\begin{theorem}[\textbf{C} has power objects]

\label{th:Closure}

If $A$ and $B$ effective sets then $A \Rrightarrow B$ with \\
$ev : El((A \Rrightarrow B) \times A) \rightarrow El(B)$ defined by
$ev(w) \eqdef app(p_1(w), p_2(w))$ is $B$ power $A$ in \textbf{C}.

Hence \textbf{C} is cartesian and closed.

\end{theorem}

\end{toappendix}

\begin{toappendix}
[
\begin{proof}
Corollary of Theorem \ref{th:ClosureProperties}.
See Appendix \thisappendix.
\end{proof}
]

\begin{proof}

$A \Rrightarrow B$ is an effective set.

By stability of $app$, $p_1$ and $p_2$ we have
$ev : El((A \Rrightarrow B) \times A) \rightarrow El(B)$

Assume $X$ is an effective set and $f : El(X \times A) \rightarrow El(B)$.

We construct $\varphi \in |El(X)| \rightarrow |El(A \Rrightarrow B)|$ defined
by $\varphi(u) \eqdef \Lambda v : A . f(<u, v>)$.

\begin{itemize}

\item If $\eqeffEl X u {u'}$:
If $\eqeffEl A v {v'}$ then $\eqeffEl {X \times A} {<u, v>} {<u', v'>}$.
So $\eqeffEl B {f(<u, v>)} {f(<u', v'>)}$.
By Theorem \ref{th:ClosureProperties} we have
$\eqeffEl {A \Rrightarrow B} {\varphi(u)} {\varphi(u')}$.
Therefore $\varphi : El(X) \rightarrow El(A \Rrightarrow B)$.

\item If $\eqeffEl {X \times A} w w$:
Then $\eqeffEl X {p_1(w)} {p_1(w)}$ and $\eqeffEl A {p_2(w)} {p_2(w)}$.
For all $v$, $v'$ in $|El(A)|$, if $\eqeffEl A v {v'}$ then
$\eqeffEl B {f(<p_1(w), v>)} {f(<p_1(w), v'>)}$.
By theorem \ref{th:ClosureProperties} we have
$\eqeffEl B {app((\Lambda v : A . f(<p_1(w), v>)), p_2(w))}
{f(<p_1(w), p_2(w)>)}$.
Hence $\eqeffEl B {app(\varphi(p_1(w)), p_2(w))} {f(<p_1(w), p_2(w)>)}$.
Moreover:
\[\begin{array}{l}
\eqeffEl {X \times A} {<p_1(w), p_2(w)>} {<p_1(w), p_2(w)>} \\
\eqeffEl X {p_1(<p_1(w), p_2(w)>)} {p_1(w)} \\
\eqeffEl A {p_2(<p_1(w), p_2(w)>} {p_2(w)}
\end{array}
\]

So $\eqeffEl {X \times A} {<p_1(w), p_2(w)>} w$.
Then, $\eqeffEl B {f(<p_1(w), p_2(w)>)} {f(w)}$.
Hence $\eqeffEl B {app(\varphi(p_1(w)), p_2(w))} {f(w)}$.
Therefore \\ $\vDash \forall w \in El(X \times A),
\eqeffEl B {app(\varphi(p_1(w)), p_2(w))} {f(w)}$.

\item Assume $\psi : El(X) \rightarrow El(A \Rrightarrow B)$ such that
$\vDash \forall w \in El(X \times A),
\eqeffEl B {app(\psi(p_1(w)), p_2(w))} {f(w)}$.

If $\eqeffEl X u u$:

\begin{itemize}

\item Then we have $\eqeffEl {A \Rrightarrow B} {\varphi(u)} {\varphi(u)}$
and $\eqeffEl {A \Rrightarrow B} {\psi(u)} {\psi(u)}$.

\item If $\eqeffEl A v v$:
Then $\eqeffEl {X \times A} {<u, v>} {<u, v>}$.
Therefore
$\eqeffEl B {app(\varphi(p_1(<u, v>)), p_2(<u, v>))} {f(<u, v>)}$ and
$\eqeffEl B {app(\psi(p_1(<u, v>)), p_2(<u, v>))} {f(<u, v>)}$.
We also have $\eqeffEl X {p_1(<u, v>)} u$ and
$\eqeffEl A {p_2(<u, v>)} v$.
So, $\eqeffEl {A \Rrightarrow B} {\varphi(p_1(<u, v>))} {\varphi(u)}$.
Hence, $\eqeffEl B {app(\varphi(p_1(<u, v>), p_2(<u, v>)}
{app(\varphi(u), v)}$.
We can also prove that
$\eqeffEl B {app(\psi(p_1(<u, v>), p_2(<u, v>))} {app(\psi(u), v}$.
Hence, $\eqeffEl B {app(\varphi(u), v)} {app(\psi(u), v)}$.

\end{itemize}

By Theorem \ref{th:ClosureProperties}, we have
$\eqeffEl {A \Rrightarrow B} {\varphi(u)} {\psi(u)}$.
Therefore $\varphi \approx \psi$.

\end{itemize}

Then we can conclude.

\end{proof}

\end{toappendix}

\subsection{The sub-object classifier}

\label{sec:SubObjectClassifier}

To prove that \textbf{C} is a topos, we only have to prove that \textbf{C} has
a sub-object classifier (which can be seen as the object of truth values).

There are several possible definition of the sub-object classifier.
The one we use here needs the notion of pullback and monomorphism which can
be used as the categorical definition of subset and defined as follows in
\textbf{C}:
\begin{definition}[Pullback of a morphism from the final object]\strut

Assume $X$ is an effective set.\\
By Theorem \ref{th:Final} we can construct $* : El(X) \rightarrow El(1)$
defined by $*(u) \eqdef <>$.

Assume $A$ and $\Omega$ are effective sets,
$True : El(1) \rightarrow El(\Omega)$ and $f : El(A) \rightarrow El(\Omega)$.\\
Then $i$ is the \emph{pullback of} $True$ along $f$ if and only if:
\begin{itemize}

\item $B$ is an effective set and $i : El(B) \rightarrow El(A)$.

\item $f \circ i \approx True \circ *$.

\item For all effective set $X$ and for all
$\varphi : El(X) \rightarrow El(A)$, if $f \circ \varphi \approx True \circ *$
then there exists a unique $\psi : El(X) \rightarrow El(A)$ (modulo $\approx$)
such that $i \circ \psi \approx \varphi$.

\end{itemize}
$B$ can be seen as the categorical definition of the subset of $A$ which
elements satisfies $f$ and with $i$ the injection morphism.
\end{definition}

\begin{definition}[Monomorphism]\strut

Assume $A$ and $B$ are effective sets and $m : El(A) \rightarrow El(B)$.

$m$ is a \emph{monomorphism} if and only if for all effective sets $X$ and for
all $f, g : El(X) \rightarrow El(A)$, if $m \circ f \approx m \circ g$ then
$f \approx g$.

This can be seen as the categorical definition of injectivity.

\end{definition}

Now we can give the definition of a sub-object classifier:

\begin{definition}[Sub-object classifier]

$\Omega$ with $True$ is a \emph{sub-object-classifier} if and only if:
\begin{itemize}

\item $\Omega$ is an effective set and $True : El(1) \rightarrow El(\Omega)$.

\item For all effective sets $A$ and for all $f : El(A) \rightarrow El(\Omega)$
there exists a pullback of $True$ along $f$.

\item For all effective sets $A$ and $B$ and for all
$m : El(A) \rightarrow El(B)$ monomorphism there exists a unique
$\chi_m : El(B) \rightarrow El(\Omega)$ (modulo $\approx$) such that
$m$ is the pullback of $True$ along $\chi_m$.

\end{itemize}
\end{definition}

As in the case of the final object, the products and the power objects,
the sub-object classifier is unique up to isomorphism.

\subsubsection{Propositions}

In this section we construct $\Omega$ and $True$ and we characterize the
property that $f \approx True \circ *$ for some
$f : El(X) \rightarrow El(\Omega)$.

The construction of the sub-object classifier is the same as in Heyting's.

\begin{definition}[Definition of $\Omega$]

We define the effective set $\Omega$ as follows:
\begin{itemize}

\item $|\Omega| \eqdef Prop$

\item $\eqeff \Omega q {q'} \eqdef q \Leftrightarrow q'$

\end{itemize}
It is trivial to prove that $\Omega$ is an effective set.

\end{definition}

Unlike the general case, $\Omega$ is an effective set where it is possible to
go directly from the high level to the low level.

\begin{definition}[Definition of prop]

If $u \in El(Prop)$ we write $prop(u) \eqdef u(\top)$

\end{definition}

And $prop$ is the opposite of $el_\Omega$ as follows:

\begin{toappendix}

\appendixbeyond 0

\begin{theorem}[Properties of $\Omega$]\strut
\label{th:PropProperties}
\begin{itemize}

\item $\vDash \forall q, q' \in Prop,
(q \Leftrightarrow q') \Leftrightarrow \eqeffEl \Omega {el_\Omega(q)}
{el_\Omega(q')}$

\item $\vDash \forall u, v \in |El(\Omega)|,
\eqeffEl \Omega u v \Rightarrow (prop(u) \Leftrightarrow prop(v))$

\item $\vDash \forall q \in Prop, prop(el_\Omega(q)) \Leftrightarrow q$

\item $\vDash \forall u, v \in El(\Omega),
(prop(u) \Leftrightarrow prop(v)) \Rightarrow \eqeffEl \Omega u v$

\end{itemize}
\end{theorem}

\end{toappendix}

\begin{toappendix}
[
\begin{proof}
Straightforward. See Appendix \thisappendix.
\end{proof}
]

\begin{proof}

\begin{itemize}

\item $q \Leftrightarrow q'$ if and only if $\eqeff \Omega q {q'}$,
if and only if $\eqeffEl \Omega {el_\Omega(q)} {el_\Omega(q')}$.

\item Assume $\eqeffEl \Omega u v$:
$prop(u)$ if and only if $u(\top)$ if and only if $v(\top)$, if and only
if $prop(v)$.

\item $prop(el_\Omega(q))$ if and only if $el_\Omega(q)(\top)$, if and only if
$q \Leftrightarrow \top$, if and only if $q$.

\item We construct $F \in (|El(\Omega)| \times |El(\Omega)|) \rightarrow Prop$
defined by: $F(u, v) \eqdef (prop(u) \Leftrightarrow prop(v)) \Rightarrow
\eqeffEl \Omega u v$.
By stability of $prop$, $F \in SP(El(\Omega), El(\Omega))$.

If $prop(el_\Omega(q)) \Leftrightarrow prop(el_\Omega(q'))$, then
$q \Leftrightarrow q'$.
So $\eqeffEl \Omega {el_\Omega(q)} {el_\Omega(q')}$.

By Theorem \ref{th:ExtensionTruth}, we have
$\vDash \forall u, v \in El(\Omega), F(u, v)$.
Then we can conclude.

\end{itemize}

\end{proof}

\end{toappendix}

Then we can connect the categorical world and the framework for the notion of
truth as follows:

\begin{definition}[Definition of $True$]\strut

We construct $True \in |El(1)| \rightarrow |El(\Omega)|$ and
defined by: $True(u) \eqdef el_\Omega(\top)$

\end{definition}

\begin{toappendix}

\appendixbeyond 0

\begin{theorem}[Properties of $True$]

\label{th:Prop}
We have $True : El(1) \rightarrow El(\Omega)$.

Furthermore, assume that $X$ is an effective set.
             
Then for all $f : El(X) \rightarrow El(\Omega)$, $f \approx True \circ *$ if and
only if $\vDash \forall u \in El(X), prop(f(u))$

\end{theorem}

\end{toappendix}

\begin{toappendix}
[
\begin{proof}
Corollary of Theorem $\ref{th:PropProperties}$. See Appendix \thisappendix.
\end{proof}
]

\begin{proof}

\begin{itemize}

\item Stability of $True$:
We have $\top \Leftrightarrow \top$.
Then $\eqeffEl \Omega {el_\Omega(\top)} {el_\Omega(\top)}$.
Hence, $\eqeffEl \Omega {True(u)} {True(u')}$.
Therefore $True : El(1) \rightarrow El(\Omega)$

\item Remark: For all $u \in |El(X)|$, $(True \circ *)(u) = el_\Omega(\top)$.

\item If $\eqeffEl X u u$ and $\eqeffEl \Omega {f(u)} {el_\Omega(\top)}$ then
$prop(f(u)) \Leftrightarrow prop(el_\Omega(\top))$.
We also have $prop(el_\Omega(\top)) \Leftrightarrow \top$.
Therefore, we have $prop(f(u))$.

\item If $\eqeffEl X u u$ and $prop(f(u))$ then $\eqeffEl \Omega {f(u)} {f(u)}$
and $\eqeffEl \Omega {(True \circ *)(u)} {(True \circ *)(u)}$.
We also have $prop(f(u)) \Leftrightarrow \top$ and
$prop(el_\Omega)(\top) \Leftrightarrow \top$.
Hence $prop(f(u)) \Leftrightarrow prop((True \circ *)(u))$.
By Theorem \ref{th:PropProperties} we have
$\eqeffEl \Omega {f(u)} {(True \circ *)(u)}$.

\end{itemize}

Then we can conclude.

\end{proof}

\end{toappendix}

If $f : El(X) \rightarrow El(\Omega)$ then the property
$f \approx True \circ *$ is often used in the characterisation of the sub-object
classifier in the topos theory.

\subsubsection{Subsets}

In this section, we prove that there exist pullbacks of $True$.

The construction of subsets in this framework is quite intuitive:
first we define the effective object that represent a subset,
then we prove the main property in the high level framework, and finally
we prove that it respects the categorical characterisation.

\begin{definition}[Definition of $\{ u \in A \mid F(u)\}$]

Assume $A$ is an effective set and \\ $F \in SP(El(A))$.
We define the effective set $\{ u \in A \mid F(u) \}$ as follow:
\begin{itemize}

\item $|\{ u \in A \mid F(u) \}| \eqdef |A|$

\item $\eqeff {\{u \in A \mid F(u)\}} x y \eqdef \eqeff A x y \wedge F(el_A(x))$

\end{itemize}
It is trivial to show that $\{ u \in A \mid F(u)\}$ is an effective set.
\end{definition}

\begin{toappendix}

\appendixbeyond 0

\begin{theorem}[Properties of $\{u \in A \mid F(u)\}$]\strut

\label{th:SubProperties}

Assume $A$ is an effective set and $F \in SP(El(A))$.
Let $B \eqdef \{ u \in A | F(u) \}$.
Then:

$\vDash \forall u, v \in |El(A)|, \eqeffEl B u v \Leftrightarrow
(\eqeffEl A u v \wedge F(u))$

\end{theorem}

\end{toappendix}

\begin{toappendix}
[
\begin{proof}
Straightforward. See Appendix \thisappendix.
\end{proof}
]

\begin{proof}

\begin{itemize}

\item Assume $\eqeffEl B u v$:

\begin{itemize}

\item Stability: If $u(a)$ and $\eqeff A a {a'}$ then $\eqeff B a a$
since $\eqeffEl B u v$.
So we have $F(el_A(a))$.
Hence $\eqeff B a {a'}$.
Therefore $u(a')$.

\item Unicity: If $u(a)$ and $u(a')$, then $\eqeff B a {a'}$
by uncicity for $u$ in $B$.
Hence $\eqeff A a {a'}$.

\item Existence and equivalence still hold (they are independent from
the equality).

\end{itemize}

Therefore $\eqeffEl A u v$.
So we have $\eqeffEl A u u$.
Then there exists $a \in |A|$ such that $\eqeffEl A u {el_A(a)}$.
Hence $u(a)$ and $\eqeffEl A {el_A(a)} u$.
So, $\eqeff B a a$.
Hence $F(el_A(a))$.
Therefore $F(u)$.

\item Assume $\eqeffEl A u v$ and $F(u)$:

\begin{itemize}

\item Stability: If $u(a)$ and $\eqeff B a {a'}$ then $\eqeff A a {a'}$.
Hence $u(a')$.

\item Unicity: If $u(a)$ and $u(a')$ then $\eqeff A a {a'}$
by unicity in $A$.
We also have $\eqeffEl A u u$.
Then $\eqeffEl A u {el_A(a)}$ by Lemma~\ref{lem:LowLevelHighLevel}.3.
Hence $F(el_A(a))$.
Therefore $\eqeff B a {a'}$.

\item The existence and equivalence still hold.

\end{itemize}

Therefore $\eqeffEl B u v$.

\end{itemize}

\end{proof}

\end{toappendix}

\begin{toappendix}

\appendixbeyond 0

\begin{theorem}[\textbf{C} has pullbacks of $True$]

\label{th:Sub}

Assume $A$ is an effective set and $f : El(A) \rightarrow El(\Omega)$.
Then $B$ defined by $B \eqdef \{ u \in A \mid prop(f(u))\}$ is an effective set 
and if $i(u) \eqdef u$ then $i : El(B) \rightarrow El(A)$ and $B$ with $i$ is
the pullback of $True$ along $f$.





\end{theorem}

\end{toappendix}

\begin{toappendix}
[
\begin{proof}
Corollary of Theorem \ref{th:SubProperties}. See Appendix \thisappendix.
\end{proof}
]

\begin{proof}

By stability of $f$ and $prop$, $B$ is well defined.

If $\eqeffEl B u {u'}$, then by Theorem \ref{th:SubProperties} we have
$\eqeffEl A u {u'}$.
Hence $\eqeffEl A {i(u)} {i(u')}$.
Therefore $i : El(B) \rightarrow El(A)$.

If $\eqeffEl B u u$ then by Theorem \ref{th:SubProperties} we have
$prop(f(u))$.
So, $prop(f(i(u)))$.
Hence $\vDash \forall u \in El(B), prop(f(i(u)))$.
Therefore, by Theorem \ref{th:Prop}, $f \circ i \approx True \circ *$.

Assume $X$ is an effective set and $\varphi : El(X) \rightarrow El(A)$ such that
$f \circ \varphi \approx True \circ *$.
Then $\vDash \forall u \in El(X), prop(f(\varphi(u)))$.
We also have $\varphi \in |El(X)| \rightarrow |El(B)|$.

We choose $\psi \eqdef \phi$.
\begin{itemize}

\item We prove $\psi : El(X) \rightarrow El(B)$:
If $\eqeffEl X u {u'}$ then $\eqeffEl A {\varphi(u)} {\varphi(u')}$.
We also have $\eqeffEl X u u$, so by hypothesis we have $prop(f(\varphi(u)))$.
Hence $\eqeffEl B {\varphi(u)} {\varphi(u')}$.
Therefore $\psi : El(X) \rightarrow El(B)$.

\item If $\eqeffEl X u u$ then $\eqeffEl A {\varphi(u)} {\varphi(u)}$.
Hence $\eqeffEl A {i(\varphi(u))} {\varphi(u)}$.
Therefore $i \circ \varphi \approx \varphi$.

\item Assume $\psi' : El(X) \rightarrow El(B)$ such that
$i \circ \psi' \approx \varphi$.
If $\eqeffEl X u u$, then $\eqeffEl A {i(\psi'(u))} {\varphi(u)}$.
So $\eqeffEl A {\psi'(u)} {\varphi(u)}$.
We also have $\eqeffEl B {\psi'(u)} {\psi'(u)}$.
So $prop(f(\psi'(u)))$.
Hence $\eqeffEl B {\psi'(u)} {\varphi(u)}$.
Therefore $\psi' \approx \varphi$.

\end{itemize}

Then we can conclude by choosing $\varphi$ as $\psi$.

\end{proof}

\end{toappendix}

\subsubsection{Monomorphisms}

We use a definition of the sub-object classifier that uses monomorphisms. So
it is useful to have a relation between the notion of monomorphisms and the
notion of injectivity (as in the category of sets) defined in the high level
of the framework as follows:

\begin{definition}[Injective]

Assume $A$ and $B$ effective sets and $f : El(A) \rightarrow El(B)$.

We say that $f$ is injective if and only if:

$\vDash \forall u, u' \in El(A), \eqeffEl B {f(u)} {f(u')} \Rightarrow
\eqeffEl A u {u'}$

\end{definition}

\begin{toappendix}

\appendixbeyond 0

\begin{theorem}[Equivalence between monorphisms and injective functions]\strut

\label{th:Mono}

Assume $A$ and $B$ effective sets and $m : El(A) \rightarrow El(B)$.

$m$ is a monomorphism in \textbf{C} if and only if $m$ is injective.

\end{theorem}

\end{toappendix}

\begin{toappendix}
[
\begin{proof}

Straightforward: We adapt the proof of the category of sets.
See Appendix~\thisappendix.
\end{proof}
]

\begin{proof}
\begin{itemize}

\item Assume $m$ is injective.
Let $X$ an effective set and $f, g : El(X) \rightarrow El(A)$ such that
$m \circ f \approx m \circ g$.
If $\eqeffEl X u u$, then $\eqeffEl A {f(u)} {f(u)}$,
$\eqeffEl A {g(u)} {g(u)}$ and $\eqeffEl B {m(f(u))} {m(g(u))}$.
By injectivity, $\eqeffEl A {f(u)} {g(u)}$.
Therefore $f \approx g$.

\item Assume $m$ is a monomorphism.
We construct the effective set $X$ defined by: \\
$X \eqdef \{ w \in A \times A \mid \eqeffEl B {m(p_1(w))} {m(p_2(w))}$.
By stability of $m$, $p_1$ and $p_2$, $X$ is well defined.

We also construct $f, g \in |El(X)| \rightarrow |El(A)|$ defined by:
$f \eqdef p_1$ and $g \eqdef p_2$.
\begin{itemize}

\item If $\eqeffEl X w {w'}$, then $\eqeffEl {A \times A} w {w'}$ by
Theorem \ref{th:SubProperties}.
So we have \\ $\eqeffEl A {p_1(w)} {p_1(w')}$.
Hence $\eqeffEl A {f(w)} {f(w')}$.
Therefore $f : El(X) \rightarrow El(A)$.

\item With the same kind of proof we also have $g : El(X) \rightarrow El(A)$.

\item If $\eqeffEl X w w$, then $\eqeffEl B {m(p_1(w))} {m(p_2(w))}$ by
Theorem \ref{th:SubProperties}.
Hence \\ $\eqeffEl B {m(f(w))} {m(g(w))}$.
Therefore, $m \circ f \approx m \circ g$.

\end{itemize}
By definition of monomorphism, $f \approx g$ (We have defined
$f,g : X \rightarrow A$ so that we do not confuse them with
$p_1, p_2 : El(A \times A) \rightarrow El(A)$.).

If $\eqeffEl A u u$, $\eqeffEl A v v$ and $\eqeffEl B {m(u)} {m(v)}$:
Then $\eqeffEl {A \times A} {<u, v>} {<u, v>}$ and
$\eqeffEl A {p_1(<u, v>)} u$.
So, $\eqeffEl B {m(p_1(<u, v>))} {m(u)}$.
We can also prove that \\ $\eqeffEl B {m(p_2(<u, v>))} {m(v)}$.
Therefore $\eqeffEl B {m(p_1(<u, v>)} {m(p_2(<u, v>))}$.
By Theorem \ref{th:SubProperties} we have $\eqeffEl X {<u, v>} {<u, v>}$.
Hence, $\eqeffEl A {f(<u, v>)} {g(<u, v>)}$ since $f \approx g$.
So, $\eqeffEl A {p_1(<u, v>)} {p_2(<u, v>)}$.
Hence, $\eqeff A u v$.

Therefore $m$ is injective.

\end{itemize}

\end{proof}

\end{toappendix}

We can notice that, to prove this theorem, we did not had to add any new
construction to the framework.

\subsubsection{Axiom of unique choice}

In this section we construct $\chi_m$.

One of the advantages of this framework as a model of high order logic is
to have the axiom of unique choice.

To express it we first define the operator of description:

\begin{definition}[Operator of description]\strut

Assume $A$ is an effective set and $F \in |El(A)| \rightarrow Prop$\quad
(\ie $F \in |El(El(A))|$).

We construct $d(F) \in |El(A)|$ defined by: for all $x \in |A|$,
$d(F)(x) \eqdef F(el_A(x))$.

$d(F)$ can be read as "the only $u$ such that $F(u)$".

\end{definition}

Before proving the axiom of unique choice:

\begin{theorem}[Properties of $d(F)$]

\label{th:AUC}

Assume $A$ is an effective set and then: \\
$\vDash \forall F \in El(El(A)), F(d(F))$
which is equivalent to:
\[\begin{array}{l}
\vDash \forall F \in |El(A)| \rightarrow Prop,\\
\strut\qquad\qquad(\forall u, u' \in |El(A)|, F(u) \Rightarrow \eqeffEl A u {u'}
\Rightarrow F(u'))\\
\strut\qquad\qquad \Rightarrow (\forall u, u' \in |El(A)|, F(u) \Rightarrow
F(u') \Rightarrow \eqeffEl A u {u'}) \\
\strut\qquad\qquad \Rightarrow
(\exists u \in |El(A)|, F(u))\\
\strut\qquad\qquad  \Rightarrow F(d(F))
\end{array}
\]
\end{theorem}

\begin{proof}

Assume we have 
\begin{itemize}
\item $\forall u, u' \in |El(A)|, F(u) \Rightarrow
\eqeffEl A u {u'} \Rightarrow F(u')$,
\item $\forall u, u' \in |El(A)|, F(u) \Rightarrow F(u') \Rightarrow
\eqeffEl A u {u'}$, and
\item
$\exists u \in |El(A)|, F(u)$.
\end{itemize}

So there exists $u \in |El(A)|$ such that $F(u)$.
By unicity of $F$ (with $u' = u$), $\eqeffEl A u u$.
\begin{itemize}

\item If $u(a)$ then by Lemma \ref{lem:LowLevelHighLevel}.3 we have
$\eqeffEl A u {el_A(a)}$.
By stability of $F$, $F(el_A(a))$.
Therefore $d(F)(a)$ by definition of $d(F)$.

\item If $d(F)(a)$, then by definition of $d(F)$ we have $F(el_A(a))$.
By unicity of $F$, $\eqeffEl A u {el_A(a)}$.
By Lemma \ref{lem:LowLevelHighLevel}.3, we have $u(a)$.

\end{itemize}
By Lemma \ref{lem:LowLevelHighLevel}.1, $\eqeffEl A u  {d(F)}$.
Therefore, by stability of $F$, we have $F(d(F))$.
\end{proof}
Theorem~\ref{th:AUC} is necessary to prove the last hypothesis needed to have a
topos.
\begin{definition}[Caracteristic morphism]\strut

Assume $A$ and $B$ is effective sets, and $m : El(A) \rightarrow El(B)$ a
monomorphism.\\
Then with define $\chi_m \in |El(B)| \rightarrow |El(\Omega)|$ defined by:
$\chi_m(v) \eqdef el_\Omega(\exists u \in El(A), \eqeffEl B {m(u)} v)$
\end{definition}

\begin{toappendix}

\appendixbeyond 0

\begin{theorem}[Monomorphisms of \textbf{C} have a caracteristic morphism]\strut

\label{th:Carac}

Assume $A$ and $B$ is effective sets, and $m : El(A) \rightarrow El(B)$ a
monomorphism.

Then we have $\chi_m : El(B) \rightarrow El(\Omega)$.

And it is the only morphism (modulo $\approx$) from $B$ to $\Omega$ in
\textbf{C} such that $m$ is the pullback of $True$ along $\chi_m$.
\end{theorem}

\end{toappendix}

\begin{toappendix}
[
\begin{proof}
We adapt the proof of the category of sets.
See Appendix~\thisappendix.\\
In particular, from a $\varphi : El(X) \rightarrow El(B)$ such that
$\chi_m \circ \varphi \approx True \circ *$ we use Theorem~\ref{th:AUC} to
construct $\psi : El(X) \rightarrow El(A)$.
\end{proof}

]

\begin{proof}

By Theorem \ref{th:Mono}, $m$ is injective.
We also have: $prop(\chi_m(v))$ if and only if there exists
$u \in |El(A)|$ such that $\eqeffEl A u u$ and $\eqeffEl B {m(u)} v$
by Theorem~\ref{th:PropProperties}.3.

Stability: If $\eqeffEl B v {v'}$ then:
If there exists $u \in El(A)$ such that $\eqeffEl A u u$ and
$\eqeffEl B {m(u)} v$, then $\eqeffEl B {m(u)} {v'}$.
Hence there exists $u \in |El(A)|$, such that $\eqeffEl A u u$ and
$\eqeffEl B {m(u)} {v'}$.
We can also prove that if
$\exists u \in El(A), \eqeffEl B {m(u)} {v'}$ then
$\exists u \in El(A), \eqeffEl B {m(u)} v$.
Therefore $\exists u \in El(A), \eqeffEl B {m(u)} v \Leftrightarrow
\exists u \in El(A), \eqeffEl B {m(u)} {v'}$.
Hence $\eqeffEl \Omega {\chi_m(v)} {\chi_m(v')}$.
Therefore $\chi_m : El(B) \rightarrow El(\Omega)$.

$m$ pullback: If $\eqeffEl A u u$ then $\eqeffEl B {m(u)} {m(u)}$.
So, there exists $u' \in |El(A)|$, such that $\eqeffEl A {u'} {u'}$
and $\eqeffEl B {m(u')} {m(u)}$.
Hence $prop(\chi_m(m(u)))$.
Therefore $\chi_m \circ m \approx True \circ *$ by Theorem~\ref{th:Prop}.

Assume $X$ is an effective set and $\varphi : El(X) \rightarrow El(A)$ such that
$\chi_m \circ \varphi \approx True \circ *$.
Then $\vDash \forall u \in El(A), prop(\chi_m(\varphi(u)))$.
For every $w \in |El(X)|$ we construct $F_w \in |El(A)| \rightarrow Prop$
defined by: $F_w(u) \eqdef \eqeffEl A u u \wedge \eqeffEl B {m(u)}
{\varphi(w)}$.
We construct $\psi \in |El(X)| \rightarrow |El(A)|$ defined by
$\psi(w) \eqdef d(F_w)$.

\begin{itemize}

\item Stability of $\psi$: If $\eqeffEl X w {w'}$:

Then $\eqeffEl X w w$.
So $prop(\chi_m(\varphi(w)))$.
Hence, there exists $u \in |El(A)|$ such that $\eqeffEl A u u$ and
$\eqeffEl B {m(u)} {\varphi(w)}$.
Therefore, there exists $u \in |El(A)|$ such that $F_w(u)$.

If $F_w(u)$ and $\eqeffEl A u {u'}$ then $\eqeffEl A {u'} {u'}$,
$\eqeffEl B {m(u)} {\varphi(w)}$ and $\eqeffEl B {m(u)} {m(u')}$.
Hence, $\eqeffEl B {m(u')} {\varphi(v)}$.
Therefore, $F_w(u')$.

If $F_w(u)$ and $F_w(u')$ then $\eqeffEl B {m(u)} {\varphi(v)}$ and
$\eqeffEl B {m(u')} {\varphi(v)}$.
So, $\eqeffEl B {m(u)} {m(u')}$.
We also have $\eqeffEl A u u$ and $\eqeffEl A {u'} {u'}$.
By injectivity, $\eqeffEl A u {u'}$.

Therefore, by Theorem \ref{th:AUC}, we have $F_w(d(F_w))$.
So, $F_w(\psi(w))$.

We can also prove that $F_{w'}(\psi(w'))$.
So, $\eqeffEl A {\psi(w')} {\psi(w')} $ and
$\eqeffEl B {m(\psi(w'))}
{\varphi(w')}$.
We also have $\eqeffEl B {\varphi(w)} {\varphi(w')}$.
Hence, $\eqeffEl B {m(\psi(w'))} {\varphi(w)}$.
Therefore, $F_w(\psi(w'))$.

By unicity of $F_w$ we have $\eqeffEl A {\psi(w)} {\psi(w')}$.
Therefore $\psi : El(X) \rightarrow El(A)$.

\item Commutation: If $\eqeffEl X w w$ then we can prove that $F_w(\psi(w))$.
Hence, $\eqeffEl B {m(\psi(w))} {\varphi(w)}$.
Therefore $m \circ \psi \approx \varphi$.

\item Unicity: Assume $\psi' : El(X) \rightarrow El(A)$ such that
$m \circ \psi' \approx \varphi$.
If $\eqeffEl X w w$ then we can prove that $F_w(\psi(w))$ and the unicity
of $F_w$.
We also have $\eqeffEl B {m(\psi'(w))} {\varphi(w)}$.
So, $F_w(\psi'(w))$.
By unicity of $F_w$ we have $\eqeffEl A {\psi(w)} {\psi'(w)}$.
Therefore $\psi \approx \psi'$.

\end{itemize}

Therefore $m$ is the pullback of $True$ along $\chi_m$.

Assume $f : El(B) \rightarrow El(\Omega)$ such that $m$ is the pullback of
$True$ along $f$.
Theorem \ref{th:Sub} gives us $C \eqdef \{ v \in B \mid prop(f(v))\}$ an
effective set and $i : El(C) \rightarrow El(B)$.
$f \circ i \approx True \circ *$ so there exists
$\varphi : El(C) \rightarrow El(A)$ such that $m \circ \varphi \approx i$.

If $\eqeffEl B v v$ then $\eqeffEl \Omega {\chi_m(v)} {\chi_m(v)}$,
$\eqeffEl \Omega {f(v)} {f(v}$ and:

\begin{itemize}

\item If $prop(\chi_m(v))$ then there exists $u \in |El(A)|$, such that
$\eqeffEl A u u$ and $\eqeffEl B {m(u)} v$.
Therefore, $\eqeffEl \Omega {f(m(u))} {f(v)}$.
Since $f \circ m \approx True \circ *$ we also have
$prop(f(m(u)))$.
Hence, $prop(f(v))$.

\item If $prop(f(v))$ then $\eqeffEl C v v$ by Theorem~\ref{th:SubProperties}.
So, $\eqeffEl A {\varphi(v)} {\varphi(v)}$ and
$\eqeffEl B {m(\varphi(v))} {i(v)}$.
Hence there exists $u \in |El(A)|$, such that $\eqeffEl A u u$ and
$\eqeffEl B {m(u)} v$.
Therefore $prop(\chi_m(v))$.

\end{itemize}

By Theorem \ref{th:PropProperties}.4 we have
$\eqeffEl \Omega {\chi_m(v)} {f(v)}$.
Therefore $\chi_m \approx f$.

\end{proof}

\end{toappendix}

By combining all the categorical results of this part we can finally conclude.

\begin{theorem}

\textbf{C} is a topos and $\Omega$ with $True$ is the sub-object classsifier.

\end{theorem}

\begin{proof}

By Theorem~\ref{th:Closure}, $C$ is Cartesian Closed.
By Theorems~\ref{th:Sub} and \ref{th:Carac}, $\Omega$ with $True$ is the
sub-object classifier in \textbf{C}
Hence \textbf{C} is a topos.

\end{proof}

\subsection{Natural numbers}

\label{sec:Nat}

The main advantage of the effective topos over the category of finite sets is
that it has an object of natural integers.

Let $\Nat$ be the set of natural integers, with $0 \in \Nat$ and
$s \in \Nat \rightarrow \Nat$ the successor function.

\begin{definition}[Equality in $\Nat$]

We construct $E \in (\Nat \times \Nat) \rightarrow Prop$ defined by:
\[\begin{array}{llll}

E(n, n) & = & \{ n \} & \\
E(n, m) & = & \emptyset & n \neq m
\end{array}
\]
\end{definition}

\begin{toappendix}

\appendixbeyond 0

\begin{lemma}[Properties of E]\strut
\label{lem:NProperties}
\begin{enumerate}

\item $\vDash E(0, 0)$

\item $\vDash \forall x, y \in \Nat, E(x, y) \Rightarrow E(s(x), s(y))$

\item $\vDash \forall x, y \in \Nat, E(x, y) \Rightarrow E(y, x)$

\item $\vDash \forall x, y, z \in \Nat, E(x, y) \Rightarrow E(y, z) \Rightarrow
E(x, z)$

\item $\vDash \forall x, y \in \Nat, E(s(x), s(y)) \Rightarrow E(x, y)$

\item $\vDash \forall x \in \Nat, E(s(x), 0) \Rightarrow \bot$

\item For all $\vDash \forall P \in \Nat \rightarrow Prop, P(0) \Rightarrow
(\forall x \in \Nat, P(x) \Rightarrow P(s(x))) \Rightarrow \\
\forall x \in \Nat, E(x, x) \Rightarrow P(x)$

\end{enumerate}

\end{lemma}

\end{toappendix}

\begin{toappendix}
[
\begin{proof}
Straightforward but we have to explicitly manipulates proofs as programs.
See Appendix~\thisappendix.
\end{proof}
]

\begin{proof}

\begin{itemize}

\item The proof of 1 is just the constant $0$.

\item The proof of 2 is the successor function.

\item The proof of 3 and 6 is the identity function.

\item The proof of 4 is just a projection.

\item The proof of 5 is the predecessor function.

\item The proof of 7 is a function defined by a simple recursion.

\end{itemize}

\end{proof}

\end{toappendix}

These properties are useful if we want to adapt this work with an extension
of system $F_\omega$.

\begin{definition}[Effective set of natural integers]

We define the effective set $Nat$ as follows:
\begin{itemize}

\item $|Nat| \eqdef \Nat$

\item $|x =_{Nat} y| \eqdef E(x, y)$

\end{itemize}
By Lemma \ref{lem:NProperties}, we can prove that $Nat$ is an effective
set.
\end{definition}
We naturally have the high level constructions as a lift of the low level ones.
\begin{definition}[Definition of $Z$ and $S$]

We define $Z \in |El(Nat)|$ and $S' \in |Nat| \rightarrow |El(Nat)|$ as
follow: $Z \eqdef el_{Nat}(0)$ and $S'(x) \eqdef el_{Nat}(s(x))$.

By Lemma \ref{lem:NProperties}.2, we have $S' : Nat \rightarrow El(Nat)$.

We write $S : El(Nat) \rightarrow El(Nat)$ the extension of $S'$.
\end{definition}
We can then prove the stability and the axioms of Peano in the high level of
the framework.

\begin{toappendix}

\appendixbeyond 0

\begin{theorem}[Properties of Nat]\strut
\label{th:NatProperties}
\begin{itemize}

\item $\vDash \eqeffEl {Nat} Z Z$

\item $\vDash \forall u, v \in |El(Nat)|, \eqeffEl {Nat} u v \Rightarrow
\eqeffEl {Nat} {S(u)} {S(v)}$

\item $\vDash \forall u, v \in El(Nat), \eqeffEl {Nat} {S(u)} {S(v)} \Rightarrow
\eqeffEl {Nat} u v$

\item $\vDash \forall u \in El(Nat), \eqeffEl {Nat} {S(u)} Z \Rightarrow \bot$

\item For all  $P \in SP(El(Nat))$,\\ if $\vDash P(Z)$ and
$\vDash \forall u \in El(Nat), P(u) \Rightarrow P(S(u))$ then
$\vDash \forall u \in El(Nat), P(u)$

\end{itemize}

\end{theorem}

\end{toappendix}

\begin{toappendix}
[
\begin{proof}
Straightforward with Lemma \ref{lem:NProperties} and Theorem
\ref{th:ExtensionTruth}.
See Appendix \thisappendix
\end{proof}
]

\begin{proof}

\begin{itemize}

\item $E(0, 0)$, so $\eqeff {Nat} 0 0$.
Hence $\eqeffEl {Nat} {el_{Nat}(0)} {el_{Nat}(0)}$.
Therefore $\eqeffEl {Nat} Z Z$.

\item By stability of $S$.

\item We construct $F \in (|El(Nat)| \times |El(Nat)|) \rightarrow Prop$ defined
by: $F(u, v) \eqdef \eqeffEl {Nat} {S(u)} {S(v)} \Rightarrow
\eqeffEl {Nat} u v$. By stability of $S$, we have $F \in SP(El(Nat), El(Nat))$.

If $\eqeff {Nat} x x$, $\eqeff {Nat} y y$ and
$\eqeffEl {Nat} {S(el_{Nat}(x))} {S(el_{Nat})(y)}$:
By definition of $S$, $\eqeffEl {Nat} {S(el_{Nat}(x))} {el_{Nat}(s(x))}$ and
$\eqeffEl {Nat} {S(el_{Nat}(y))} {el_{Nat}(s(y))}$.
Hence $\eqeffEl {Nat} {el_{Nat}(s(x))} {el_{Nat}(s(y))}$.
Therefore, $E(s(x), s(y))$.
By Lemma \ref{lem:NProperties}.5, $E(x, y)$.
Hence $\eqeffEl {Nat} {el_{Nat}(x)} {el_{Nat}(y)}$.

By Theorem \ref{th:ExtensionTruth} we have
$\vDash \forall u, v \in El(Nat), F(u, v)$.
Then we can conclude.

\item We construct $F \in |El(Nat)| \rightarrow Prop$ defined by:
$F(u) \eqdef \eqeffEl {Nat} {S(u)} Z \Rightarrow \bot$.
By stability of $S$ we have $F \in SP(El(Nat))$.

If $\eqeff {Nat} x x$ and $\eqeffEl {Nat} {S(el_{Nat}(x))} Z$:
By definition of $S$, $\eqeffEl {Nat} {S(el_{Nat}(x))} {el_{Nat}(s(x))}$.
Hence $\eqeffEl {Nat} {el_{Nat}(s(x))} {el_{Nat}(0)}$.
Therefore $E(s(x), 0)$.
By Lemma \ref{lem:NProperties}.6 we have $\bot$.

By Theorem \ref{th:ExtensionTruth} we have $\vDash \forall u \in El(Nat), F(u)$.
Then we can conclude.

\item We construct $P' \in |Nat| \rightarrow Prop$ defined by:
$P'(x) \eqdef E(x, x) \wedge P(el_{Nat}(x))$.

\begin{itemize}

\item $E(0, 0)$ and $P(el_{Nat}(0))$ so we have $P'(0)$.

\item If $P'(x)$ then $E(x, x)$ and $P(el_{Nat}(x))$.
Hence $P(S(el_{Nat}(x)))$.
We also have $\eqeffEl {Nat} {S(el_{Nat}(x))} {el_{Nat}(s(x))}$.
Hence $E(s(x), s(x))$ and $P(el_{Nat}(s(x)))$.
Therefore $P'(s(x))$.

\end{itemize}

By Lemma \ref{lem:NProperties}.5, $\vDash \forall x \in \Nat, E(x, x)
\Rightarrow
P'(x)$.
Hence $\vDash \forall x \in Nat, P(el_{Nat}(x))$.
By Theorem \ref{th:ExtensionTruth} we have
$\vDash \forall u \in El(Nat), P(u)$.

\end{itemize}

\end{proof}

\end{toappendix}

Finally we can conclude than \textbf{C} has an object of natural integers:

\begin{theorem}[\textbf{C} has an object of integers]

We construct $Z_m \in |El(1)| \rightarrow |El(Nat)|$ defined by
$Z_m(u) \eqdef Z$. $(Nat, Z_m, S)$ is the object of natural integers in the
category \textbf{C} which means that:
\begin{itemize}

\item $Z_m : El(1) \rightarrow El(Nat)$ and
$S : El(Nat) \rightarrow El(Nat)$

\item For all $X$ effective set, $f : El(1) \rightarrow El(X)$,
$g : El(X) \rightarrow El(X)$, there exists a unique (modulo $\approx$)
$\varphi : El(Nat) \rightarrow El(X)$ such that
$\varphi \circ Z_m \approx f$ and $\varphi \circ S \approx g \circ \varphi$

\end{itemize}
\end{theorem}

\begin{proof}

With Theorem \ref{th:NatProperties}, we can prove that $(Nat, Z_m, S)$
satisfies the Peano axioms in the internal logic of the topos \textbf{C}.
Therefore, $(Nat, Z_m, S)$ is the object of naturals integers in \textbf{C}.

\end{proof}

Because our framework is based on realisability, we can do program extraction:

\begin{theorem}[Program extraction]

\label{th:ProgramExtraction}

Assume $f : El(Nat) \rightarrow El(Nat)$.

There exists a unique $g \in \Nat \rightarrow \Nat$ such that:

\[\vDash \forall x \in Nat, \eqeffEl {Nat} {f(el_{Nat}(x))} {el_{Nat}(g(x))}\]

And then, $g$ is computable.

\end{theorem}

\begin{proof}

If $E(x, x)$ then $\eqeffEl {Nat} {el_{Nat}(x)} {el_{Nat}(x)}$.
So, $\eqeffEl {Nat} {f(el_{Nat}(x))} {f(el_{Nat}(x))}$.
Hence, there exists $m \in |Nat|$, such that
$\eqeffEl {Nat} {f(el_{Nat}(x))} {el_{Nat}(m)}$.\\
Therefore, $\vDash \forall x \in \Nat, E(x, x) \Rightarrow \exists m \in N,
\eqeffEl {Nat} {f(el_{Nat}(x))} {el_{Nat}(m)}$.

Hence, there exists 
$e \in \forall x \in \Nat, E(x, x) \Rightarrow \exists m \in \Nat,
\eqeffEl {Nat} {f(el_{Nat}(x))} {el_{Nat}(m)}$.
Then $\varphi_e$ exists.
Let $x \in \Nat$.
So $x \in E(x, x)$.
Hence $\varphi_e(x) \downarrow \in \exists m \in \Nat, \eqeffEl {Nat}
{f(el_{Nat}(x))} {el_{Nat}(m)}$.
Therefore, there exists $m \in \Nat$ such that
$\varphi_e(x) \in \eqeffEl {Nat} {f(el_{Nat}(x))} {el_{Nat}(m)}$.
Let $m' \in \Nat$ such that $\varphi_e(x) \in \eqeffEl {Nat}
{f(el_{Nat}(x))} {el_{Nat}(m')}$.
Hence $\vDash \eqeffEl {Nat} {f(el_{Nat}(x))} {el_{Nat}(m)}$ and
$\vDash \eqeffEl {Nat} {f(el_{Nat}(x))} {el_{Nat}(m')}$.
Therefore $\vDash E(m, m')$.
So there exists $k \in E(m, m')$.
Hence $m = m'$.
Therefore there exists a unique $g : \Nat \rightarrow \Nat$ such that for all
$x \in \Nat$, $\varphi_e(x) \downarrow \in \eqeffEl {Nat} {f(el_{Nat}(x))}
{el_{Nat}(g(x))}$.
Hence $e \in \forall x \in Nat, \eqeffEl {Nat}
{f(el_{Nat}(x))} {el_{Nat}(g(x))}$.
Therefore $\vDash \forall x \in Nat,
\eqeffEl {Nat} {f(el_{Nat}(x))} {el_{Nat}(g(x))}$.

Assume $h : \Nat \rightarrow \Nat$ such that $\vDash \forall x \in Nat,
\eqeffEl {Nat} {f(el_{Nat}(x))} {el_{Nat}(h(x))}$.
Hence $\vDash \forall x \in \Nat, E(x, x) \Rightarrow E(g(x), h(x))$.
So, there exists $e \in \forall x \in \Nat, E(x, x) \Rightarrow
E(g(x), h(x))$.
Therefore, $\varphi_e$ exists, and for all $x \in \Nat$,
$\varphi_e(x) \downarrow \in E(g(x), h(x))$.
Hence $\varphi_e(x) = g(x) = h(x)$.
Therefore $g = h$.

If we choose $g = h$, then there exists $e$ such that $\varphi_e$ exists,
and for all $x \in \Nat$, $\varphi_e(x) \downarrow = g(x)$.
Therefore $g$ is computable.

\end{proof}
This theorem itself motivates the construction of effective topos and it
allows us to construct computable functions from our framework.







\section{Conclusion}

We have built a realisability framework and with a different but equivalent
definition of the Effective Topos, we have proved that the Effective Topos
was a topos with an object of natural integers simply by using the high-level
part of this framework
(Theorems \ref{th:FinalProperties}, \ref{th:ProductProperties},
\ref{th:ClosureProperties}, \ref{th:PropProperties}, \ref{th:SubProperties} and
\ref{th:AUC}) and adapting the proof that the category of sets is a topos.
The only difference is that when constructing a morphism, we have to
check stability which is always straightforward with the stability
properties.
Moreover, most of the construction of the high-level of the framework is
facilitated by Theorems~\ref{th:ExtensionTruth}
and ~\ref{th:ExistenceUnicityExtensionFun}.

With this framework we can manipulates algebraic types because we have integers
and the power of Topos Theory.
But it would be better to have the algebraic types as a core feature of the
framework.

As a future work we could make a typing system of high order logic where
the syntax would be trivially inspired of the high-level part of the framework:
the framework would be a trivial model of this system and this would prove
the correctness of the system and the ability of extracting proofs.
We should use an extended version which manipulates for example dependent types
which could be integrated in a future version of this framework.

Therefore, without having any knowledge about category or topos theory,
this framework can be a solid ground for future work in an intuitionistic higher-order logic system with interesting properties on the equality.


\bibliography{Common/abbrev-short,Common/Main,Common/crossrefs,Bib}

\appendix

\section{Full proofs}

\gettoappendix {lem:LowLevelHighLevel}
\gettoappendix {lem:LowLevelHighLevelproof}
\gettoappendix {th:ExistenceUnicityExtensionFun}
\gettoappendix {th:ExistenceUnicityExtensionFunproof}
\gettoappendix {lem:SoundComp}
\gettoappendix {lem:SoundCompproof}
\gettoappendix {th:FinalProperties}
\gettoappendix {th:FinalPropertiesproof}
\gettoappendix {th:Final}
\gettoappendix {th:Finalproof}
\gettoappendix {lem:ProductPreStable}
\gettoappendix {lem:ProductPreStableproof}
\gettoappendix {th:ProductProperties}
\gettoappendix {th:ProductPropertiesproof}
\gettoappendix {th:Product}
\gettoappendix {th:Productproof}
\gettoappendix {lem:PreClosure}
\gettoappendix {lem:PreClosureproof}
\gettoappendix {th:ClosureProperties}
\gettoappendix {th:ClosurePropertiesproof}
\gettoappendix {th:Closure}
\gettoappendix {th:Closureproof}
\gettoappendix {th:PropProperties}
\gettoappendix {th:PropPropertiesproof}
\gettoappendix {th:Prop}
\gettoappendix {th:Propproof}
\gettoappendix {th:SubProperties}
\gettoappendix {th:SubPropertiesproof}
\gettoappendix {th:Sub}
\gettoappendix {th:Subproof}
\gettoappendix {th:Mono}
\gettoappendix {th:Monoproof}
\gettoappendix {th:Carac}
\gettoappendix {th:Caracproof}
\gettoappendix {lem:NProperties}
\gettoappendix {lem:NPropertiesproof}
\gettoappendix {th:NatProperties}
\gettoappendix {th:NatPropertiesproof}

\section{Comparing with other categories}

\subsection{Variant with strict morphisms}

In this subsection we are going to wonder what happens if we had a condition of
strictness on the morphisms.

First we define some kind of Klop construction to our framework:

\newcommand{\Klop}[2]{[{#1} \mid {#2}]}

\begin{definition}[Definition of the Klop Construction]

Assume $X$ is an effective set, $u \in |El(X)|$ and $F \in Prop$.
We	write $\Klop u F \in |El(X)|$ defined by:
${\Klop u F}(x) \eqdef u(x) \wedge F$

\end{definition}

\begin{lemma}[Properties of the Klop Construction]

\label{lem:KlopProperties}

Assume $X$ is an effective set. Then:

\[\vDash \forall u, v \in |El(X)|, F \in Prop,
\eqeffEl X u {\Klop v F} \Leftrightarrow (\eqeffEl X u v \wedge F)
\]

\end{lemma}

\begin{proof}

\begin{itemize}

\item If $\eqeffEl X u {\Klop v F}$:
Then, there exists $x \in |X|$ such that $u(x)$.
Hence, we have ${\Klop v F}(x)$.
Therefore $F$.
We also have $\eqeffEl X u u$.
If $u(y)$ then ${\Klop v F}(y)$, so $v(y)$.
If $v(y)$ then ${\Klop v F}(y)$, so $u(y)$.
Therefore $\eqeffEl X u v$.

\item If $\eqeffEl X u v$ and $F$:
Then $\eqeffEl X u u$.
If $u(x)$ then $v(x)$, so ${\Klop v F}(x)$.
If ${\Klop v F}(x)$ then $v(x)$, so $u(x)$.
Therefore $\eqeffEl X u {\Klop u F}$.

\end{itemize}

\end{proof}

\begin{definition}[Definition of $C_{strict}$]

We defined the category $C_{strict}$ as follows:

\begin{itemize}

\item The objects of $C_{strict}$ are the effective sets.

\item The morphisms from $X$ to $Y$ in $C_{strict}$ are the
$f : El(X) \rightarrow El(Y)$ (modulo $\approx$) such that:

\[\vDash \forall u \in |El(X)|, \eqeffEl Y {f(u)} {f(u)} \Rightarrow
\eqeffEl X u u
\]

\item The composition is the usual composition

\end{itemize}

$C_{strict}$ is indeed a category.

\end{definition}

We write:

\begin{itemize}

\item $F(X) \eqdef G(X) \eqdef X$

\item $F(f)(u) \eqdef [f(u) | |u =_{El(X)} u|]$

\item $G(f) \eqdef f$

\end{itemize}

\begin{lemma}[Properties of $F$ and $G$]\strut

\label{lem:StrictProperties}

\begin{itemize}

\item For all $f : El(X) \rightarrow El(Y)$,
$\vDash \forall u \in El(X), |(F f)(u) = f(u)|$

\item $F$ is a functor from \textbf{C} to $C_{strict}$

\item $G$ is a functor from $C_{strict}$ to \textbf{C}

\item $G \circ F = Id_{C}$ and $F \circ G = Id_{C_{strict}}$

\end{itemize}

\end{lemma}

\begin{proof}

The first point is a corollary of Lemma \ref{lem:KlopProperties}.
Then the other points are trivial.

\end{proof}

\begin{theorem}

\label{th:Strict}

\textbf{C} and $C_{strict}$ are isomorph, hence equivalent.

\end{theorem}

\begin{proof}

Corollary of Lemma \ref{lem:StrictProperties}.

\end{proof}

So having strict morphisms does not change the power of the effective topos.
It is just more complicated to use.

\subsection{The usual effective topos}

\label{sec:Usual}

\begin{definition}[Definition of $C_{usual}$]

We write $C_{usual}$ the usual definition of the effective topos which is
the following:

\begin{itemize}

\item The objects of $C_{usual}$ are the effective sets.

\item The morphisms from $X$ to $Y$ in $C_{usual}$ are the
$F \in (|X| \times |Y|) \rightarrow Prop$ such that:

\begin{itemize}

\item $\vDash \forall x \in |X|, y \in |Y|,
F(x, y) \Rightarrow (\eqeff X x x \wedge \eqeff Y y y)$

\item $\vDash \forall x, x' \in |X|, y, y' \in |Y|,
F(x, y) \Rightarrow \eqeff X x {x'} \Rightarrow \eqeff Y y {y'} \Rightarrow
F(x', y')$

\item $\vDash \forall x, x' \in |X|, y, y' \in |Y|,
F(x, y) \Rightarrow F(x', y') \Rightarrow \eqeff X x {x'} \Rightarrow
\eqeff Y y {y'}$

\item $\vDash \forall x \in |X|, \eqeff X x x \Rightarrow \exists y \in |Y|,
F(x, y)$

\end{itemize}

Modulo the relation $\approx$ defined by:

\[F \approx G \eqdef (\vDash \forall x \in |X|, y \in |Y|,
F(x, y) \Leftrightarrow G(x, y))\]

\item If $F \in (|X| \times |Y|) \rightarrow Prop$ and
$G \in (|Y| \times |Z|) \rightarrow Prop$ then \\
$G \circ F \in (|X| \times |Z|) \rightarrow Prop$ is defined as follows:

$$(G \circ F)(x, z) \eqdef \exists y \in |Y|, F(x, y) \wedge G(y, z)$$

This notion of composition is coherent with $\approx$.

\end{itemize}

\end{definition}

We write:

\begin{itemize}

\item $\Phi(X) \eqdef \Psi(X) \eqdef X$

\item $\Phi(f)(x, y) \eqdef \eqeffEl Y {f(el_X(x))} {el_Y(y)}$

\item $\Psi(F)(u)(y) \eqdef \exists x \in |X|, \eqeffEl X u {el_X(x)} \wedge
F(x, y)$

\end{itemize}

\begin{lemma}[Properties of $C_{usual}$]

\label{lem:UsualProperties}

\begin{itemize}

\item $\Phi$ is a functor from $C_{strict}$ to $C_{usual}$.

\item $\Psi$ is a functor from $C_{usual}$ to $C_{strict}$.

\item $\Psi \circ \Phi = Id_{C_{strict}}$ and $\Phi \circ \Psi =
Id_{C_{usual}}$.

\end{itemize}

\end{lemma}

\begin{proof}

Straightforward.
In particular, we use the strictness of $f$ to prove the strictness of
$\Phi(f)$.

\end{proof}






\begin{theorem}

\textbf{C} and $C_{usual}$ are isomorph, hence they are equivalent.

\end{theorem}

\begin{proof}

By Lemma \ref{lem:UsualProperties}, $C_{strict}$ and $C_{usual}$ isomorph.
By Theorem \ref{th:Strict}, $\textbf{C}$ and $C_{strict}$ isomorph.
Therefore $\textbf{C}$ and $C_{usual}$ isomorph.

\end{proof}

This legitimates the fact that we call $C$, the category we constructed, by the
name of "effective topos".

\subsection{Naive variant}

\label{sec:Naive}

The category $C_{naive}$ is the naive definition of \textbf{C} without using the
high level tools ($El(X)$, etc ...).

\begin{definition}[Definition of $C_{naive}$]

We define the category $C_{naive}$ as follows:

\begin{itemize}

\item The objects of $C_{naive}$ are the effective sets.

\item The morphisms from $X$ to $Y$ in $C_{naive}$ are the
$f : X \rightarrow Y$ (modulo $\approx$).

\item The composition is the usual composition.

\end{itemize}

$C_{naive}$ is indeed a category.

\end{definition}

With $C_{naive}$, we can adapt most of the work we have done in part 3 with
\textbf{C}.
Except that we cannot prove the axiom of unique choice.
Hence, $C_{naive}$ is not a topos and it would have been a bad choice to choose
it to base our framework on it.

Of course \textbf{C} is not equivalent to $C_{naive}$.

\end{document}